\documentclass[journal]{IEEEtran}

\usepackage[latin1]{inputenc}
\usepackage{graphicx}
\usepackage{amssymb,amsmath,amsfonts,amsthm}
\usepackage{algorithm,algorithmic}
\usepackage{cite}
\usepackage{float}
\usepackage{graphics}
\usepackage{subfigure}
\usepackage{xspace}
\usepackage[usenames,dvipsnames]{color}
\usepackage{epsfig}
\usepackage{mathtools}

\newtheorem{theorem}{Theorem}
\newtheorem{proposition}{Proposition}
\newtheorem{remark}{Remark}

\begin{document}

\title{Relay Selection in Wireless Powered Cooperative Networks with Energy Storage}

\author{Ioannis Krikidis,~\IEEEmembership{Senior Member,~IEEE}
\thanks{I. Krikidis is with the Department of Electrical and Computer Engineering, Faculty of Engineering, University of Cyprus, Nicosia 1678 (E-mail: {\sf krikidis@ucy.ac.cy}).}
\thanks{This work was supported by the Research Promotion Foundation, Cyprus,
under the Project FUPLEX with Pr. No. CY-IL/0114/02.}}

\maketitle

\begin{abstract}
This paper deals with the problem of relay selection in wireless powered cooperative networks, where spatially random relays are equipped with energy storage devices e.g., batteries. In contrast to conventional techniques and in order to reduce complexity, the relay nodes can either harvest energy from the source signal (in case of uncharged battery) or attempt to decode and forward it (in case of charged battery). Several relay selection schemes that correspond to different state information requirements and implementation complexities are proposed. The charging/discharging behavior of the battery is modeled as a two-state Markov chain and analytical expressions for the steady-state distribution and the outage probability performance are derived for each relay selection scheme. We prove that energy storage significantly affects the performance of the system and results in a zeroth diversity gain at high signal-to-noise ratios; the convergence floors depend on the steady-state distribution of the battery and are derived in closed-form by using appropriate approximations. The proposed relay selection schemes are generalized to a large-scale network with multiple access points (APs), where relays assist the closest AP  and suffer from multi-user interference.   
\end{abstract}


\section{Introduction}

\IEEEPARstart{T}{he} roll-out of the {\it Internet of Things} will lead to the massive deployment of sensor nodes and a vast amount of information exchange, making it impractical, even impossible, to individually recharge/control these devices on a regular basis. Wireless powered communication (WPC) is a promising energy solution for the future highly dense and heterogeneous networks with major impact on many different applications. It refers to communication networks where particular nodes power their operations by the received electromagnetic radiation.  The fundamental block for the implementation of this technology is the rectifying-antenna (rectenna) which is a diode-based circuit that converts the radio-frequency (RF) signals to direct-current voltage \cite{BOR,KIM}.

WPC appears in the literature in  three basic architectures that refer to different application scenarios; a rigorous survey of WPC is given in \cite{Bi15,Lu15}. The first main architecture is the wireless power transfer (WPT), where a dedicated RF transmitter wirelessly transfers power to devices \cite{XU,zeng15}. In contrast to ambient RF harvesting, WPT can be continuous and fully controlled and therefore is attractive for applications with strict quality-of-service constraints. The second network architecture is the wireless powered communication network (WPCN), where a dedicated RF transmitter broadcasts energy at the downlink and WPT-based devices transmit information at the uplink \cite{JU,CHEN}. The third fundamental architecture is the simultaneous wireless information and power transfer (SWIPT), where the RF transmitter simultaneously conveys data and energy at the downlink devices \cite{ZHA,MORSI}. Due to practical constraints, SWIPT cannot be performed from the same signal without losses and practical implementations split the received signal in two parts, where one part is used for information transfer and another part is used for power transfer. This signal split can be performed in \cite{KRI1} i) time domain i.e., time switching (TS),  ii) power domain i.e., power splitting (PS), or iii) spatial domain i.e., antenna switching (AS). 

A special case of SWIPT with particular interest for the cooperative networks is the SWIPT with energy/information relaying. In this network structure, a batteryless relay node extracts both information and energy from the source signal, and then uses the harvested energy to forward the source signal to a destination. In \cite{NAS1}, the authors study the performance of a three-node Amplify-and-Forward (AF) relay channel, where the relay node employs TS/PS to power the relaying link. This work is extended in \cite{NAS2} for a Decode-and-Forward (DF) relay channel and the throughput performance is analyzed in closed form for both TS/PS techniques. A three-node relay channel with direct link, which combines TS-SWIPT with the dynamic Decode-and-Forward (DDF) protocol, is analyzed in \cite{KOJ}. The work in\cite{KRI2} studies a multiple-input multiple-output (MIMO) relay channel, where the relay node employs AS technique for SWIPT; the optimal AS is formulated as a binary knapsack problem and low-complexity solutions are proposed. Recent studies re-examine  the fundamental three-node SWIPT relay channel by assuming that the relay node has full-duplex (FD) radio capabilities \cite{ZHO,ZEN}. In \cite{ZHO}, the authors study the optimal beamforming for a FD MIMO relay channel where a relay node uses TS/PS in order to power the relaying link; an appropriate beamforming vector protects the relay input from the loop-channel effects. The integration of the loop-channel interference into the WPT process is discussed in \cite{ZEN}, and the optimal power allocation and beamforming design at the relay node are derived.       

In contrast to the above studies which focus on the basic three-node relay channel, 
more complicated network structures appear in the literature. A network topology where multiple source-destination pairs communicate through a common DF relay node that employs PS-SWIPT is discussed in \cite{DIN1}. In \cite{WAN}, the authors study a two-way relay channel where multiple communication pairs exchange information through a  shared PS-SWIPT MIMO relay node; the optimal beamforming vector that satisfies some well-defined fairness criteria is investigated.  The work in \cite{HE} introduces multiple SWIPT relays and studies the interference relay channel, when multiple sources communicate with their destinations through dedicated SWIPT relays. On the other hand, the authors in \cite{DIN2} study the relay selection problem when an access point (AP) communicates with the destination through multiple PS-SWIPT relay nodes, which are randomly located in the space. This work analyzes the outage probability performance for three main relay selection techniques: i) random relay selection, ii) relay selection based on the Euclidean distance, and iii) distributed (cooperative) beamforming. The consideration of the spatial randomness in WPC studies is of vital importance, since RF harvesting highly depends on the path-loss e.g., \cite{DIN2,HAM,RUI2} use stochastic geometry tools in order to study different WPC networks. 

A SWIPT-based relay  extracts both information and energy from the source transmission. For single-antenna relays, the TS technique requires perfect time synchronization and suffers from interrupted information transmission, since dedicated time slots are used for power transfer. On the other hand, the PS technique dynamically splits the received signal at the relay node and requires appropriate circuits that increase the implementation complexity. The implementation of PS/TS-SWIPT relay channel is an open research problem and several practical issues should be resolved.    

On the other hand, the integration of an energy storage device (e.g., battery, capacitor etc.) at the relay nodes, which is charged by the received RF radiation, introduces another WPC-based relay structure with new potentials and challenges; this network architecture has not been fully explored. The work in \cite{KRI3} introduces a three-node relay channel where the relay node is equipped with a multi-level battery that can be charged by the source signal; the optimal switching policy between WPT and information decoding at the relay node is investigated. This work is extended in \cite{KUA} for a scenario with multiple relay nodes where a single-best relay node is selected for relaying; however, this work does not take into account spatial randomness, which is critical for WPC networks. A WPC network with spatial randomness and batteries is studied in \cite{RUI2}; specifically, the authors study the performance of a large-scale cognitive radio network, where cognitive nodes charge their batteries from the primary transmissions. 

In this paper, we focus on a WPC cooperative network and we study the problem of relay selection, when relay nodes are equipped with batteries and are randomly located in the coverage area of an AP. We assume that the relay nodes hold two-state storage devices (e.g., batteries, capacitors) which are charged by the source transmission; relay nodes become active and can assist source transmission, only when are fully charged. Based on this setup, we re-design the relay selection schemes proposed in \cite{DIN2} (for conventional batteryless relays) and we study their performance. We investigate several relay selection policies that have different channel state information (CSI) requirements and correspond to different complexities. Analytical results for the outage probability performance of the proposed schemes as well as simplified asymptotic expressions for the high signal-to-noise ratio (SNR) regime are derived. We prove that in contrast to the conventional case \cite{DIN2}, where a diversity gain equal to one is ensured, the battery model results in a zeroth diversity gain at high SNRs. The associated outage probability floor becomes lower when the battery status is taken into account with regards to the  relay selection decision. Furthermore, the relay selection schemes are generalized to a large-scale network with randomly located APs, where multi-user interference affects relay decoding. To the best of the authors' knowledge, the problem of relay selection for WPC cooperative networks with batteries at the relay nodes, has not been reported in the literature. 

\noindent {\it Paper organization:} We introduce the system model and the main assumptions in Section \ref{sys_model}. In Section \ref{SEC2}, we present the relay selection policies and we analyze their outage probability performance. Section \ref{SEC3} generalizes the relay selection schemes for a multi-cell network. Simulation results are presented in Section \ref{NR}, followed by our conclusions in Section \ref{SEC5}.

\noindent {\it Notation:} $\mathbb{R}^d$ denotes the $d$-dimensional Euclidean space, $|z|$ denotes the magnitude of a complex variable $z$, $\mathbb{P}(X)$ denotes the probability of the event $X$, $\mathbb{E}(\cdot)$ represents the expectation operator; $\gamma(a,x)$ denotes the lower incomplete gamma function \cite[Eq. 8.350]{GRAD}, and $_2F_1(a,b;c;x)$ is the Gaussian or ordinary hypergeometric function \cite[Eq. 9.100]{GRAD}.

\section{System model}\label{sys_model}

We consider a single-cell where an AP communicates with a destination, $D$, via the help of a set of relay nodes $R_i$, whose locations are determined from a Poisson point process (PPP). All nodes are equipped with a single antenna and the relay nodes  have WPT capabilities. The AP is backlogged and transmits with a fixed power $P$ and a spectral efficiency $r_0$ bits per channel use (BPCU);  the AP's transmitted signal is the only WPT source for the relay nodes. Fig. \ref{model} schematically presents the system model.

\subsection{Topology}
We consider a disc, denoted by $\mathcal{D}$, where the AP is located at the origin of the disc and the radius of the disc is $\rho$.  The relay nodes are located inside the disc and their locations form a homogeneous PPP $\Phi$ (inside $\mathcal{D}$) of density $\lambda$; $N$ denotes the number of the relays. The distance between the AP and the destination is denoted by $d_0$ and there is not a direct link between the AP and the destination e.g., due to obstacles or severe shadowing. We assume that $d_i$ denotes the Euclidean distance between AP and the $i$-th relay $R_i$ with $1 \leq i \leq N$, while $c_i$ denotes the distance between $R_i$ and $D$.  

\subsection{Cooperative protocol}
The relay nodes have half-duplex capabilities and employ a DF policy. Time is slotted and communication is performed in two orthogonal time slots. Specifically, in the first time slot i.e., broadcast phase, the AP broadcasts the signal to the relay nodes. In the second time slot, a single relay (or a group of relays), which successfully decoded the source signal, forwards the signal to the destination according to the rules of the considered DF-based relay selection scheme. It is worth noting that the selected relay transmits at a rate $r_0$ (same with the rate of the source).

\begin{figure}
\centering
\includegraphics[width=0.8\linewidth]{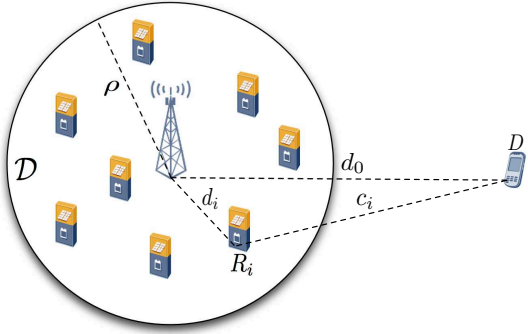}\\
\vspace{-0.4cm}
\caption{The network model; an AP communicates with the destination $D$ via the help of relays.}\label{model}
\end{figure}

\subsection{Battery model}\label{batt}

Each relay node is equipped with a single energy storage device (e.g., battery, capacitor, etc.), which can store energy for future use.  At the beginning of the broadcast phase, the battery can be either fully charged or empty (two states) \cite[Sec. II.A]{RUI2}; a relay with a charged battery is active and can participate in the relaying operation, while a relay with an empty battery is in harvesting mode and uses the AP's signal for WPT purposes. We assume that the battery can be used either for charging or transceiver operations (a single battery/capacitor cannot be charged and discharged simultaneously \cite{RUI3}). Fig. \ref{battery} schematically shows the battery model and the associated switching mechanism.   

The WPT process is based on the AP's transmissions and WPT from relaying signals is negligible i.e., relay nodes transmit in a much lower power than the AP and thus cannot satisfy the rectennas' sensitivity requirements. {Since the battery has only two discrete states (full or empty),  an empty battery can be fully charged when the input power is larger than the size of the battery. Let $P_r=\Psi P$ be the size/capacity of the battery with $\Psi<1$. The battery of the selected relay (if it is charged) is connected to the transceiver's circuit and is fully discharged at the end of the second/relaying time slot\footnote{This assumption ensures that the state-transition probability matrix is independent of the channel conditions (path-loss, fading) and simplifies the analysis.}. If the decoding is successful, most of the available energy is used for transmission i.e., the relay node transmits with a fixed power $P_r$ (transmission  dominates the total energy consumption). If the decoding is not successful, the available energy is mainly consumed in order to operate the basic receiver's components  (e.g., RF electronics, signal processing, etc.), transmit a Negative-Acknowledgment signal to indicate unsuccessful decoding as well as for static maintenance operations (e.g., cooling system).

\begin{figure}
\centering
 \includegraphics[width=0.7\linewidth]{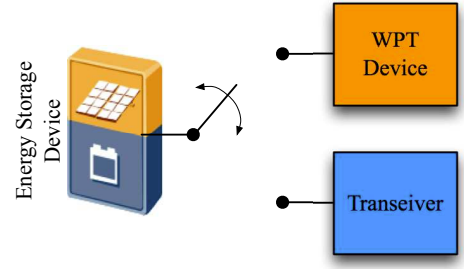}\\
\vspace{-0.4cm}
\caption{The battery model; the battery is either connected to the WPT device for charging or to the transceiver device for discharging. }\label{battery}
\end{figure}

\noindent{\it Markov chain model:} The charging/discharging behavior of the battery can be represented by a finite-state Markov chain (MC) with two states $\{s_0,s_1\}$, where the state $s_0$ indicates that the battery is empty and the state $s_1$ that the battery is charged. For the sake of presentation, we introduce the function $S(R_i)$, which returns the state of the battery for the $i$-th relay i.e., $S(R_i)=s_j$ means that the battery of the $i$-th relay is in the $s_j$ state, where $j\in\{0,1\}$. The state-transition probability matrix $\pmb{\Pi}$ can be written as
\begin{align}
\pmb{\Pi}=\left[ \begin{array}{cc} 1-\pi_0 & \pi_0\\ \pi_1 & 1-\pi_1   \end{array} \right],
\end{align}  
where $\pi_0$ denotes the probability that the input power is greater than the battery size of the relay, and $\pi_1$ is the probability that a charged relay node is selected for relaying. Both probabilities are defined in the following section for the proposed relay selection schemes. If $\pmb{\eta}=[\eta_0\;\eta_1]$  denotes the stationary steady-state probability vector of the MC, by using basic queueing theory, we have

\begin{align}
\pmb{\eta}\pmb{\Pi}=\pmb{\eta}. \label{first_eq}
\end{align}
The solution of the above system of linear equations gives the steady-state distribution of the battery, which is equal to
\begin{align}
\eta_0=\frac{\pi_1}{\pi_0+\pi_1},\;\;\eta_1=\frac{\pi_0}{\pi_0+\pi_1}. \label{sst}
\end{align}

\noindent {\it Extension to a multi-state battery model:} If a relay node becomes active only when is fully charged,  a battery model with $L+2$ states follows the same analysis with the two-state case; it only affects the computation of the steady-state distribution.  Specifically,  the battery of size $P_r$ is discretized in $L+2$ energy levels $\epsilon_i=i P_r/(L+1)$ with $i=0,\ldots, L+1$. In this case, the harvesting threshold is smaller than the size/capacity of the battery and is equal to the difference between two successive energy levels i.e., $P_r/(L+1)$. We define $L +2$ corresponding energy states, $s_i$, with $i =0, 1, \ldots,L +1$ and thus the battery is in state $s_i$ when its stored energy is equal to $\epsilon_i$. If $p_{i,j}$ denotes the  transition probability $\mathbb{P}\{s_i \rightarrow s_j\}$, the state-transition probability matrix $\pmb{\Pi}$ has a dimension $(L+2)\times (L+2)$. A steady state $\pmb{\eta}$ exists and can be calculated by solving the linear system in \eqref{first_eq}. We note that the probabilities $p_{i,j}$ with $i\neq L+2$ depend on the channel statistics and can be calculated using the procedure in \cite{KRI3}. On the other hand, $p_{L+2,j}=0$ for $0< j<L+2$, while the transition probabilities  $p_{L+2,0}$, $p_{L+2,L+2}$ depend on the relay selection schemes and follow the same analytical steps with the two-state case.

\subsection{Channel model}

We assume that wireless links suffer from both  small-scale block fading and large-scale path-loss effects.   The fading is Rayleigh distributed so the power of the channel fading is an exponential random variable with unit variance. We denote by $h_i$ and $g_i$, the channel coefficients for the links between AP and $R_i$, $R_i$ and $D$, respectively. The path-loss model assumes that the received power is proportional to $1/(1+d^{\alpha})$ where $d$ is the Euclidean distance between the transmitter and the receiver, $\alpha>2$ denotes the path-loss exponent and we define $\delta\triangleq 2/\alpha$. 
The considered path-loss model ensures that the path-loss is always larger than one for any distance i.e., $1+d^\alpha>1$, even if $d<1$. The instantaneous fading channels are known only at the receivers, except if otherwise defined. In addition, all wireless links exhibit additive white Gaussian noise (AWGN) with variance $\sigma^2$; $n_i$ denotes the AWGN at the $i$-th node.  

The proposed relay selection schemes are analyzed in terms of outage probability i.e., the probability that the destination cannot support the target spectral efficiency. If $C(x)=\frac{1}{2}\log_2(1+x)$ denotes the instantaneous capacity for  a wireless link (one-hop transmission) with SNR $x$, the associated outage probability is given by $\mathbb{P}\{C(x)<r_0 \}=\mathbb{P}\{x<\epsilon \}$ where $\epsilon\triangleq 2^{2r_0}-1$; to simplify the notation, we define $\Xi\triangleq \epsilon \sigma^2 /P$.

\section{Relay selection and performance analysis}\label{SEC2}
In this section,  we present several relay selection schemes that correspond to different complexities. The outage probability performance as well as the diversity gain of the proposed schemes are derived in closed form.

\subsection{Random relay selection}\label{RRS}

The random relay selection (RRS) scheme does not require any feedback about the battery status or the location of the relay nodes and selects a relay in a random way. It corresponds to a low implementation complexity and is appropriate for networks with strict power/bandwidth constraints. Without loss of generality,  we consider that the $i$-th relay is selected to assist the source. If the $i$-th relay is fully charged, it attempts to decode the source signal and acts as a relay in case of successful detection. If the $i$-th relay has an empty battery, it switches to the harvesting mode and uses the received signal for WPT purposes; in this case, the relay remains inactive during the cooperative slot and an outage event occurs. On the other hand, the non-selected relays with empty batteries switch to harvesting mode and use the source signal for potential charging. During the first time slot, the received signal at the $i$-th relay can be written as
\begin{align}
y_i=\sqrt{P}\frac{h_i}{\sqrt{1+d_i^\alpha}}s+n_i, \label{m1}
\end{align}
where $s$ denotes the source signal with normalized power. If the $i$-th relay is active, the associated received SNR is equal to 
\begin{align}
{\sf SNR}_i=\frac{P|h_i|^2}{(1+d_i^\alpha)\sigma^2}.
\end{align}
\noindent If the $i$-th relay is inactive (empty battery), the input power at the WPT device is equal to
\begin{align}
P_h=\zeta\frac{P|h_i|^2}{(1+d_i^\alpha)}, \label{harv1}
\end{align}
where $\zeta$ denotes the WPT conversion efficiency and \eqref{harv1} assumes that energy harvesting from AWGN is negligible (very small compared to the desired signal \cite{ZHA}). In our analysis we assume $\zeta=1$ without loss of generality.   In the second time slot, the received signal at the destination can be written as 
\begin{align}
y_D=\sqrt{P_r}\frac{g_i}{\sqrt{1+c_i^\alpha}}s+n_D, 
\end{align}
and the associated SNR is given by 
\begin{align}
{\sf SNR}_D=\frac{P_r|g_i|^2}{(1+c_i^\alpha)\sigma^2}. \label{dest1}
\end{align}

In order to derive the outage probability of the RRS protocol, firstly, we need to calculate the steady-state distribution of the battery. 

\begin{proposition}\label{thr1}
The steady-state probability that a relay node is charged at the beginning of the broadcast phase, for the RRS scheme, is given by
\begin{align}
\eta_1^{\emph{RRS}}=\frac{\delta \exp \left(-\Psi\right)\frac{\gamma(\delta, \Psi \rho^\alpha)}{\Psi^{\delta}}}{\delta \exp \left(-\Psi\right)\frac{\gamma(\delta, \Psi \rho^\alpha)}{\Psi^{\delta}}+\frac{1}{\lambda \pi}}. \label{eta1}
\end{align}
\end{proposition}

\begin{proof}
See Appendix \ref{APB}.
\end{proof}

The expression in \eqref{eta1} shows that for highly dense networks $\lambda\rightarrow \infty$, the steady-state probability that a relay is charged approaches to $1$. This was expected, since as the density of the network increases, the relay selection probability approaches to zero ($1/N\rightarrow 0$) and therefore 
the relay nodes are most of the time in harvesting mode. On the other hand, the steady-state probability decreases as $\Psi$ increases ($\eta_1^{\text{RRS}}$ is a decreasing function of $\Psi$); as the harvesting threshold  increases (size of the battery), the probability to have an input power higher than the threshold decreases. The spectral efficiency does not affect the battery status of the relays, since uncharged relays observe only the energy content of the received signals, while charged relays are fully discharged  in case of selection (independently of the decoding status (success or failure)).

As for the outage probability performance of the RRS scheme, an outage event occurs when a) there are not any available relays in $\mathcal{D}$ i.e., $N=0$, b) the disc contains at least one relay ($N\geq 1$) but the selected relay has an empty battery, c) $N\geq 1$, the selected relay is fully charged but cannot decode the source signal, and d) $N\geq 1$, the selected relay is fully charged and decodes the source signal but the destination cannot support the targeted spectral efficiency. By analyzing the probability of these events, the outage probability for the RRS scheme is given by the following theorem. 

\begin{theorem}\label{thr2}
The outage probability achieved by the RRS scheme is given by \eqref{outag1}. 
\end{theorem}
\begin{proof}
See Appendix  \ref{APC}.
\end{proof}

\begin{table*}
\begin{align}
\Pi_{\text{RRS}}=& \exp(-\lambda \pi \rho^2)\!+\!\big(1-\exp(-\lambda \pi \rho^2)\big)\!\Bigg[1-\eta_1^{\text{RRS}} \underbrace{\frac{\delta}{\pi \rho^4}\frac{\gamma(\delta, \Xi \rho^\alpha)}{\Xi^{\delta}}\exp \left(-\Xi \left(1+\frac{1}{\Psi} \right)\right)\int_{0}^{2\pi}\int_{0}^{\rho}\exp\left(-\frac{\Xi}{\Psi} (x^2+d_0^2-2x d_0\cos (\theta))^\frac{1}{\delta} \right)x dx d\theta}_{\triangleq Q} \Bigg]. \label{outag1}
\end{align}
\end{table*}

Theorem \ref{thr2} can be used to study the diversity gain of the RRS scheme, as shown in the following remark.

\begin{remark}\label{cor1}
For the special case with $P\rightarrow \infty$, $P_r\rightarrow \infty$, $\Psi=P_r/P$ (constant ratio), $\rho<<d_0$, and $N\geq 1$, the outage probability of the RRS scheme is given by the expression
\begin{subequations}
\begin{align}
\Pi_{\emph{RRS}}^{\infty}&\approx 1-\eta_1^{\emph{RRS}}\left[1-\Xi\left(\frac{1+d_0^\alpha}{\Psi}+1\right)  \right] \label{as1} \\
&\rightarrow 1-\eta_1^{\emph{RRS}}. \label{skopelaki}
\end{align}
\end{subequations}
\end{remark}
\begin{proof}
When $N\geq 1$, the first factor in  \eqref{outag1} can be ignored, while $c_i\approx d_0$ due to the assumption $\rho<<d_0$. For $P\rightarrow \infty$ (e.g., $\Xi\rightarrow 0$), and by using the approximations $1-\exp(-x)\approx x$ and $\gamma(a,x)\approx x^a/a$ for  $x\rightarrow 0$, the expression in Remark \ref{cor1} can be obtained in a straightforward  way. This remark assumes $N\geq 1$ in order to highlight the impact of the battery on the achieved outage probability performance; it also holds for $N\geq 0$ for moderate/high values of $\lambda \pi \rho^2$. It is worth noting that the asymptotic expression in \eqref{skopelaki} is general and holds for any $d_0$.
\end{proof}
The above expression shows that the outage probability for the RRS scheme suffers from an outage floor at high SNRs, which depends on the steady-state distribution of the battery.

\subsection{Relay selection based on the closest distance}

The relay selection based on the closest distance (RCS) requires an a priori knowledge of the location of the relay nodes. We assume that the AP monitors the location of the relays through a low-rate feedback channel or a Global Positioning System mechanism, and selects the relay node which is closest to the AP.  The RCS scheme does not take into account battery status and/or instantaneous fading and also corresponds to a low implementation complexity,  specifically for scenarios with low mobility.  The relay node selected by the RCS policy can be expressed as
\begin{align}
R^*=\arg_{R_i \in \Phi}\min_{i=1,\ldots,N} d_i.
\end{align}    
The mathematical description of the RCS scheme follows the equations in  Section \ref{RRS}, by simply replacing the random relay node $R_i$ with $R^*$. The steady-state distribution of the battery for the RCS scheme is given by the following proposition.

\begin{proposition}\label{cor2}
The steady-state probability that a relay node is charged at the beginning of the broadcast phase,  for the RCS protocol, is given by
\begin{align}
\eta_1^{\emph{RCS}}=\eta_1^{\emph{RRS}}. 
\end{align}
\end{proposition}
\begin{proof}
Both RRS and RCS schemes select a single relay for transmission without taking into account the battery status. In addition, the selected relay is discharged at the end of the relaying slot independently of its decoding efficiency. Since both schemes handle the selected relay in the same way and the probability to select a relay in the RCS scheme is also $\mathbb{E}[1/N]\approx 1/\lambda \pi \rho^2$ (i.e., a relay can be the closest with the same probability), the steady-state distribution of the RCS scheme is equivalent to the RRS scheme.     
\end{proof}

For the outage probability performance, the scenarios for an outage event follow the discussion in Section \ref{RRS}. The following theorem provides an exact expression for the outage probability achieved by the RCS scheme. 

\begin{theorem}\label{thr3}
The outage probability achieved by the RCS scheme is given by  \eqref{oo2}.
\end{theorem}

\begin{table*}
\begin{align}
\Pi_{\text{RCS}}&= \exp(-\lambda \pi \rho^2)+\bigg(1-\exp(-\lambda \pi \rho^2)\bigg) \nonumber \\
&\;\;\times\Bigg[ 1-\eta_1^{\text{RCS}} \underbrace{\frac{2\pi\lambda^2 \exp\left(-\Xi(1+\frac{1}{\Psi}) \right)}{[1-\exp\big(-\pi \lambda \rho^2 \big)]^2}\int_{0}^{2\pi } \int_{0}^{\rho} \int_{0}^{\rho}\exp\left(-\frac{\Xi}{\Psi} (r^2+d_0^2-2r d_0\cos(\theta))^{\frac{1}{\delta}}-\lambda\pi r^2 -\Xi x^\alpha-\lambda \pi x^2 \right) r x dr dxd\theta}_{\triangleq Q'(\lambda)} \Bigg]. \label{oo2} 
\end{align}
\end{table*}

\begin{proof}
See Appendix \ref{APD}.
\end{proof}
In order to simplify the outage expression of the RCS scheme and study the diversity gain of the system, we provide the following remark.

\begin{remark}\label{cor3}
For the special case with $P\rightarrow \infty$, $P_r\rightarrow \infty$, $\Psi=P_r/P$ (constant ratio), $\rho<<d_0$, $N\geq 1$ and $\alpha=2$, the outage probability of the RCS scheme is given by 
\begin{subequations}
\begin{align}
\Pi_{\emph{RCS}}^{\infty}&\approx 1-\eta_1^{\emph{RCS}}\frac{\lambda\pi}{\lambda \pi+\Xi} \left(1+\Xi\rho^2 \frac{\exp(-\lambda \pi \rho^2)}{1-\exp(-\lambda \pi \rho^2)} \right) \nonumber \\
&\;\;\;\;\times \left[1-\Xi\left(\frac{1+d_0^2}{\Psi}+1\right)  \right], \label{as2} \\
&\rightarrow 1-\eta_1^{\emph{RCS}}. \label{skopelaki2}
\end{align}
\end{subequations}
\end{remark}
\begin{proof}
See Appendix \ref{APD2}.
\end{proof}

The above remark shows that the outage probability of the RCS scheme converges to an outage floor at high SNRs, which depends on the charging behavior of the battery.  By comparing the expressions in Remarks \ref{cor1} and \ref{cor2}, we can see that both schemes converge to the same outage floor and therefore become asymptotically equivalent (the convergence floor is independent of $\alpha$ and $d_0$). However, by carefully comparing the expressions in \eqref{as1} with $\alpha=2$ and \eqref{as2}, it can be seen that the RCS scheme converges to the outage floor faster than the RRS scheme.

\subsection{Random relay selection with battery information}

The random relay selection with battery information (RRSB) scheme randomly selects a relay node among the charged relays (if any). The RRSB scheme is based on a priori knowledge of the battery status and requires relays to feed their battery status ($1$-bit feedback) at the beginning of each broadcast phase\footnote{ Distributed implementation can be also considered where synchronized local timers at the relays (which could depend on the geographical location) allow the charged relays to access the channel according to the considered relay selection policy \cite{BLE}.}. The steady-state distribution of the battery is given as follows

\begin{proposition}
The steady-state probability that a relay node is charged at the beginning of the broadcast phase, for the RRSB scheme, is given by
\begin{align}\label{aprjens}
\eta_1^{\emph{RRSB}}=1-\frac{\Psi^\delta}{\delta \lambda \pi \exp(-\Psi)\gamma(\delta,\Psi\rho^\alpha)},
\end{align}
\end{proposition}
\begin{proof}
The charged relays form a PPP $\Omega$ which yields from the original PPP $\Phi$ by applying a thinning operation; its density is equal to $\lambda_\Omega=\lambda \eta_1^{\text{RRSB}}$; let $N'$ denote the number of the charged relays in the disc, which follows a Poisson distribution with parameter $\lambda_\Omega$.   
The RRSB scheme randomly selects a single relay from the PPP $\Omega$ and thus the  probability that a fully charged battery becomes empty is equal to the selection probability i.e., 
\begin{align}
\pi_1^{\text{RRSB}}=\mathbb{E}\left[\frac{1}{N'} \right]\geq \frac{1}{\mathbb{E}[N']}=\frac{1}{\lambda_\Omega \pi \rho^2}=\frac{1}{\lambda \eta_1^{\text{RRSB}}\pi \rho^2}, \label{appp}
\end{align}
where the expression in \eqref{appp} applies Jensen's inequality.  On the other hand, the probability that an empty battery becomes fully charged ($\pi_0$) is similar to the RRS scheme and follows the analysis in Appendix \ref{APB}. By substituting $\pi_0$ and $\pi_1^{\text{RRSB}}$ into \eqref{sst} and solving the linear equation for $\eta_1^{\text{RRSB}}$, we have $\eta_1^{\text{RRSB}}=\frac{\pi_0-\frac{1}{\lambda \pi \rho^2} }{\pi_0}$ with the required condition $\pi_0\geq 1/\lambda \pi \rho^2$.

It is worth noting that the proposed approximation allows the derivation of a simple closed-form expression for the steady-state distribution and becomes more efficient for $\lambda \pi \rho^2>>0$. The efficiency of the proposed approximation is discussed in Section \ref{NR}. 
\end{proof}

By combining the expressions in \eqref{eta1} and \eqref{appp}, we can show that $\eta_1^{\text{RRS}}=1/(2-\eta_1^{\text{RRSB}})\geq \eta_1^{\text{RRSB}}$. This observation is justified by the fact that the RRSB scheme limits the selection only between the charged relays and thus a charged relay can be discharged with a higher probability $\pi_1$.

As for the outage probability performance, an outage event occurs when a) there is not any charged relay in the disc, or b) the first or the second hop of the relay transmission is in outage. By analyzing these events, the following theorem is given.   

\begin{theorem}
The outage probability achieved by the RRSB scheme is given by the expression in
\begin{align}
\Pi_{\emph{RRSB}}\!=\!\exp(-\lambda_\Omega \pi \rho^2)+(1-\exp(-\lambda_\Omega \pi \rho^2))(1-Q), \label{RRSB}
\end{align}
where $Q$ is the success probability for the relaying link and is defined in \eqref{outag1}.
\end{theorem}

\begin{proof}
Based on the above outage events, the outage probability for the RRSB scheme is written as 
\begin{equation}
\begin{split}
\Pi_{\text{RRSB}} & =\mathbb{P}\{N'=0\}+\mathbb{P}\{N'\geq 1 \} \nonumber \\
&\;\;\times [1-\mathbb{P}\{{\sf SNR}_i\geq \epsilon|N'\geq 1\}\mathbb{P}\{{\sf SNR}_D\geq \epsilon|N'\geq 1 \}] \nonumber \\
& =\exp(-\lambda_\Omega \pi \rho^2)+(1-\exp(-\lambda_\Omega \pi \rho^2))(1-Q). \nonumber 
\end{split}
\end{equation}

\end{proof}
When $P, P_r \rightarrow \infty$ with $\Psi=P_r/P$ (constant ratio), we can straightforwardly show that  $Q\rightarrow 1$ and thus the outage probability is dominated by the event where no relay is fully charged in the disc; the outage probability in \eqref{RRSB} asymptotically converges to the following outage floor 
\begin{align}
\Pi_{\text{RRSB}}^{\infty}\rightarrow\exp(-\lambda_{\Omega} \pi \rho^2)=\exp(-\lambda \eta_1^{\text{RRSB}}\pi \rho^2). \label{memes}
\end{align}

\subsection{Relay selection based on the closest distance with battery information}
The relay selection based on the closest distance with battery information (RCSB) scheme follows the principles of the RCS scheme, but takes into account the battery status of the relays nodes. Specifically, the RCSB scheme selects the closest charged relay according to 
\begin{align}
R^*=\arg_{R_i\in \Omega} \min_{i=1,\ldots,N'} d_i.
\end{align}
As far as the steady-state probability of the battery is concerned, the charging/discharging behavior of the batteries follows the discussion of the RRSB scheme.  

\begin{proposition}
The steady-state probability that a relay node is charged at the beginning of the broadcast phase, for the RCSB scheme, is given by
\begin{align}
\eta_1^{\emph{RCSB}}=\eta_1^{\emph{RRSB}}. 
\end{align}
\end{proposition}

By straightforwardly applying the analysis of the RCS scheme for the outage events discussed in the RRSB scheme, the outage probability for the RCSB scheme is given by the following theorem. 

\begin{theorem}
The outage probability achieved by the RCSB scheme is given by the expression
\begin{align}
\Pi_{\emph{RCSB}}=\exp(-\lambda_{\Omega} \pi \rho^2)+(1-\exp(-\lambda_{\Omega} \pi \rho^2))[1- Q'(\lambda_{\Omega})], \label{RCSB2}
\end{align}
where $Q'(\cdot)$ is the success probability for the relaying link defined in \eqref{oo2}.
\end{theorem}
When $P, P_r \rightarrow \infty$ with $\Psi=P/P_r$ (constant ratio),  by using the same arguments with \eqref{memes}, we can show that the outage probability in \eqref{RCSB2} asymptotically converges to the following outage floor 
\begin{align}
\Pi_{\text{RCSB}}^{\infty}\rightarrow\exp(-\lambda_{\Omega} \pi \rho^2)=\exp(-\lambda \eta_1^{\text{RCSB}}\pi \rho^2). \label{jkos}
\end{align}

\subsection{Distributed beamforming}

The distributed beamforming (DB) scheme selects all the relay nodes that are charged at the beginning of the broadcast phase; it is mainly used as a useful performance benchmark for the single-relay selection schemes.  More specifically, all the relay nodes with fully charged batteries  become active and attempt to decode the source signal.   The relays that are able to successfully decode the source signal, create a virtual multiple antenna array and coherently transmit the signal to the destination (virtual multiple-input single-output channel). The practical implementation of the distributed beamforming requires the encoding of the source message with a  cyclic  redundancy  check  (CRC)  code  for  error detection. In this way, only the relays whose CRCs check transmit in the second  phase of the protocol. The DB scheme requires a perfect time synchronization and signaling between the relay nodes as well as a CSI at the relays; a feedback channel ensures the knowledge of the $g_i$ channel coefficient at the $i$-th relay. The broadcast phase of the DB protocol follows the description of the RRS and RCS schemes and is given by the expressions in \eqref{m1}-\eqref{harv1}. The received signal at the destination, during the second phase of the protocol, can be written as

\begin{align}
y_D=\sqrt{P_r} \sum_{i \in \mathcal{C}}\frac{w_i g_i}{\sqrt{1+c_i^\alpha}}s+n_D,  \label{yd}
\end{align}
where $w_i=g_i^*/\sqrt{\sum_{i \in \mathcal{C}}|g_i|^2}$ is the precoding coefficient at the $i$-th relay that ensures coherent combination of the relaying signals at the destination, and $\mathcal{C}$ denotes the set of the charged relay nodes (index), which successfully decoded the source message and participate in the relaying transmission \cite{DIN2}. The associated SNR at the destination is given by 
\begin{align}
{\sf SNR}_D=P_r \sum_{i\in\mathcal{C}}\frac{|g_i|^2}{(1+c_i^\alpha)\sigma^2}.\label{memoides}
\end{align}

Equivalently to the previous protocols, the outage probability performance of the DB scheme depends on the probability of a battery to be charged at the beginning of the broadcast phase. Towards this direction, we state the following proposition. 

\begin{proposition}\label{thr4}
The steady-state probability that a relay node is charged at the beginning of the broadcast phase,  for  the DB protocol, is given by
\begin{align}
\eta_1^{\emph{DB}}=\frac{\delta \exp \left(-\Psi\right)\frac{\gamma(\delta, \Psi \rho^\alpha)}{\Psi^{\delta}}}{\delta\exp \left(-\Psi\right)\frac{\gamma(\delta, \Psi \rho^\alpha)}{\Psi^{\delta}}+\rho^2}. \label{expmem}
\end{align}
\end{proposition}
\begin{proof}
The behavior of an empty battery follows the discussion of the RRS and RCS schemes and therefore $\pi_0$ is given by the expression in \eqref{po}. On the other hand, the proposed DB scheme enforces a fully charged battery to be discharged at the end of the relaying slot,  since any charged relay is selected for potential transmission. According to the considered two-state battery model, any charged selected relay is discharged at the end of the relaying slot whether participating in the relaying transmission (successful decoding) or not. This means that the transition probability $\pi_1$ in this case is equal to one. By plugging \eqref{po} and $\pi_1=1$ into $\eta_1$ of \eqref{sst}, Proposition \ref{thr4} is proven. 
\end{proof}

The expression in \eqref{expmem} shows that the steady-state probability is independent of the density $\lambda$. The DB scheme enforces all the charged relays to be fully discharged at the end of the relaying time slot and thus the size of the network does not affect the battery status distribution.

Based on the description of the DB scheme, an outage event occurs when a) the relay set $\mathcal{C}$ is empty or b) when the coherent relaying transmission is in outage i.e., the destination is not able to decode the relaying signal. For the outage probability performance of the DB scheme, we state the following theorem. 

\begin{theorem}\label{Thr5}
The outage probability achieved by the DB scheme when $\rho<<d_0$, is given by
\begin{align}
\Pi_{\emph{DB}}&=\sum_{k=0}^{\infty} \frac{\gamma \left(k,\frac{\Xi(1+d_0^\alpha)}{\Psi} \right)}{\Gamma(k)} \exp(-\lambda' \pi \rho^2)  \frac{(\lambda' \pi \rho^2)^k}{k!}, \label{thr55}
\end{align}
where $\lambda'=\lambda \eta_1^{\emph{DB}}\frac{\delta}{\rho^2}\exp \left(-\Xi\right)\frac{\gamma(\delta, \Xi \rho^\alpha)}{\Xi^{\delta}}$.
\end{theorem}

\begin{proof}
See Appendix \ref{APE}.
\end{proof}
It is worth noting that the sum in \eqref{thr55} quickly converges to the outage probability of the system and only a small number of terms is required (i.e., less than $10$ terms).  For the high SNR regime, we provide the following simplified expression, given as a remark.  

\begin{remark}\label{cor4}
For the special case with $P\rightarrow \infty$, $P_r\rightarrow \infty$, $\Psi=P_r/P$ (constant ratio), the outage probability of the DB scheme is given by 
\begin{align}
\Pi_{\emph{DB}}^{\infty}& \approx \exp(-\lambda' \pi \rho^2) I_0 \left(2\rho \sqrt{ \frac{\Xi(1+d_0^\alpha)\lambda' \pi }{\Psi} } \right)\nonumber \\
& \rightarrow \exp(-\lambda \eta_1^{\emph{DB}} \pi \rho^2). \label{jkos2}
\end{align}
\end{remark}

\begin{proof}
See Appendix \ref{APF}.
\end{proof}
The expression in Remark \ref{cor4} shows that for high SNRs, the outage probability of the DB scheme is equal to the probability that the set $\mathcal{C}$ is empty. This probability is an exponential function of $\eta_1^{\text{DB}}$ and approaches zero as
$\lambda \eta_1^{\text{DB}}\rho^2$ increases However, the DB scheme corresponds to a higher complexity, since it requires a continuous feedback to enable coherent combination of the relaying signals at the destination.    

\begin{remark}
For the special case with $P\rightarrow \infty$, $P_r\rightarrow \infty$, $\Psi=P_r/P$ (constant ratio) and $\lambda \pi \rho^2>>0$, the outage probabilities of the proposed relay selection schemes are ordered as
\begin{equation}
\Pi_{\rm{RCSB}}^{\infty}=\Pi_{\rm{RRSB}}^{\infty}<\Pi_{\rm{DB}}^{\infty}<\Pi_{\rm{RCS}}^{\infty}=\Pi_{\rm{RRS}}^{\infty}.
\end{equation}
\end{remark}

\begin{proof}
The proof follows straightforwardly from the expressions derived in \eqref{skopelaki}, \eqref{skopelaki2}, \eqref{memes}, \eqref{jkos} and \eqref{jkos2}. Regarding the battery-based selection schemes, which result in an exponential outage probability, we have $\eta_1^{\text{DB}}=\frac{1}{1+\nu \rho^2}<\eta_1^{\text{RRSB}}=1-\frac{\nu}{\lambda \pi}\Rightarrow \exp(-\lambda \eta_1^{\text{RRSB}} \pi \rho^2 )< \exp(-\lambda \eta_1^{\text{DB}}\pi \rho^2)$ with $\nu=\frac{\Psi^\delta}{\delta \exp(-\Psi)\gamma(\delta,\Psi\rho^\alpha)}<\lambda \pi$;  due to path-loss attenuation, practical implementations consider small $\Psi$ and thus this condition is satisfied with strict inequality.  The intuition behind this result is that in the asymptotic regime,  where the achieved outage probability floor is a function of $\eta_1$ and independent of the channel,  the selection of a single relay ensures connectivity while keeps more charged relays (and thus a higher $\eta_1$) than the DB scheme. 
\end{proof}

\section{Generalization to multi-cell scenarios}\label{SEC3}
In this section, we discuss the generalization of the proposed relay selection schemes for large-scale networks with multiple APs. We focus on the single-relay selection schemes (RRS, RCS, RRSB, RCSB) and we study their behavior when multi-user interference affects the decoding operation at the relay nodes (well-known model in the literature e.g., \cite{NUW,HIMA1}).  More specifically, we assume a multi-cell network, where APs form a PPP $\Upsilon$ with density $\mu$ on the plane $\mathbb{R}^2$. The relay nodes are also spatially distributed on the plane $\mathbb{R}^2$ according to a PPP with density $\lambda$. Each  relay  is connected to the closest AP and therefore all the relay nodes, which are located inside the same Voronoi cell, are dedicated to assist the corresponding AP. We study a typical cell, $\mathcal{D}$, where the AP is located at the origin (Slyvnyak's Theorem \cite{HAN}). For tractability, we assume that each Voronoi cell can be approximated by a disc of radius $\rho=1/4\sqrt{\mu}$ \cite{KOU} and each AP has a single destination at Euclidean distance $\rho$ in some random direction (worst case scenario). A direct link between AP and destination is not available and communication can be performed only through the DF relay nodes.  In addition, we assume that each relay node transmits in an orthogonal channel and thus destinations are free from multi-user interference. This setup is inline with modern network architectures, where the relay nodes (e.g., femtocells, distributed antenna systems, etc.) have cognitive radio capabilities and thus can opportunistically access the channel. By using appropriate sensing radio mechanisms, the relay nodes exploit unoccupied spectrum holes or white spaces in order to minimize/mitigate interference \cite{LI,KPO}. The topology considered is an extension of the single-cell case and therefore the system model follows the discussion in Section  \ref{sys_model} (except if some parameters are defined otherwise); Fig. \ref{gmodel} depicts the network topology for the multi-cell case. 

\begin{figure}[t]
\centering
\includegraphics[width=0.6\linewidth]{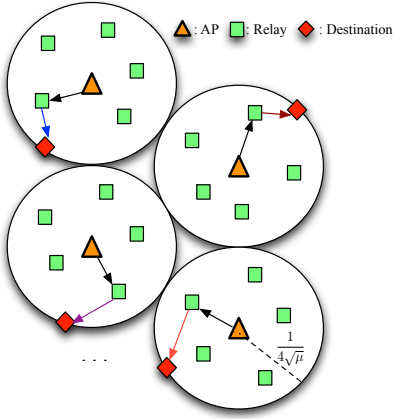}\\
\vspace{-0.3cm}
\caption{Network topology for the multi-cell scenario; Voronoi cells are approximated by discs of radius $1/4\sqrt{\mu}$ and each relay transmits in a dedicated channel.}\label{gmodel}
\end{figure}

\subsection{RRS scheme}   
According to the principles of the RRS scheme, each AP selects a relay node that is inside its Voronoi cell, in a random way. Since WPT highly depends on the Euclidean distance, we assume that WPT from external APs is negligible i.e., each relay node harvests energy only from the associated AP. Therefore, the multi-cell scenario does not change the analysis for the steady-state distribution of the battery and the probability that a relay node is charged, is given by Proposition \ref{thr1} with $\rho=1/4\sqrt{\mu}$. Let $R_i$ denote the selected relay for the typical cell; the signal to interference and noise ratio (SINR) at the $i$-th relay is given by 
\begin{align}
{\sf SINR}_i=\frac{\frac{P |h_i|^2}{d_i^\alpha}}{P I+\sigma^2},
\end{align}
where $I\triangleq \sum_{j\in \Upsilon/\{\mathcal{D} \}} \frac{H_j}{r_j^\alpha}$ denotes the aggregate (normalized) multi-user interference at the typical relay node,  $H_j$ denotes the channel power for the link between the $j$-th interfering  AP and the selected relay, and $r_j$ denotes the associated Euclidean distance. It is worth noting, that for the broadcast phase of the cooperative protocol, the path-loss function $1+d^\alpha$ is replaced by the conventional unbounded model $d^\alpha$ \cite{HAN}; this assumption significantly simplifies the analysis and allows the derivation of closed form expressions. For the second phase of the cooperative protocol, since relay nodes transmit in orthogonal channels, the SNR at the destination is given by \eqref{dest1}. For the outage probability of the RRS scheme, we state the following theorem.

\begin{theorem}\label{Thr6}
The outage probability of the RRS scheme for the multi-cell case is given by
\begin{align}
\Pi_{\emph{RRS}_0}=\Lambda(\lambda,1/4\sqrt{\mu}),
\end{align}
where $\Lambda(\cdot)$ is defined in \eqref{outag111}.
\end{theorem}
\begin{proof}
See Appendix \ref{APH}.
\end{proof}

\begin{table*}
\begin{align}
\Lambda(\lambda, \rho)=& \exp(-\lambda \pi \rho^2)+\big(1-\exp(-\lambda \pi \rho^2)\big)\Bigg[1-\eta_1^{\text{RRS}}\underbrace{\frac{2}{\rho^2}\int_{0}^{\rho}\exp\left(-\Xi x^\alpha \right)\exp\left(-\frac{\pi}{16}\left[\frac{\epsilon x^\alpha}{\rho^{\alpha}}\frac{_2F_1\left(1,2\;; 2-\delta\;; \frac{1}{1+\frac{\rho^\alpha}{\epsilon x^\alpha}}\right)}{(1-\delta)\left(\frac{\epsilon x^\alpha}{\rho^\alpha}+1 \right)^2}\!-\!\frac{\epsilon x^\alpha}{\epsilon x^\alpha+\rho^\alpha}\right]  \!\right) x dx}_{\triangleq Q_1(\Xi, \rho)} \nonumber \\
&\times \frac{1}{\pi \rho^2}\exp\left(-\frac{\Xi}{\Psi}\right) \int_{0}^{2\pi}\int_{0}^{\rho}\exp\left(-\frac{\Xi}{\Psi} (r^2+\rho^2-2r \rho\cos (\theta))^\frac{1}{\delta} \right)r dr d\theta\Bigg], \label{outag111}
\end{align}
\end{table*}

For the high SNR regime, the outage probability of the RRS scheme can be simplified as follows

\begin{remark}
For the special case with $P\rightarrow \infty$, $P_r\rightarrow \infty$, $\Psi=P_r/P$ (constant ratio) and $N\geq 1$, the outage probability of the RRS scheme is given by
\begin{align}
\Pi_{\emph{RRS}_0}^{\infty}\rightarrow 1-\eta_1^{\emph{RRS}}Q_1(0,1/4\sqrt{\mu})
\end{align}
where $Q_1(\cdot)$ is defined in \eqref{outag111}.
\end{remark}
In comparison to the single-cell case, we can see that the multi-user interference affects the first hop of the relaying protocol and therefore the achieved outage probability converges to a higher outage floor i.e., for the same radius $\rho$, $\Pi_{\text{RRS}}^{\infty}=1-\eta_{1}^{\text{RRS}}<\Pi_{\text{RRS}_0}^{\infty}=1-\eta_{1}^{\text{RRS}} Q_1(0,\rho)$, where $Q_1(0,\rho)<1$ (see \eqref{outag111}) is the probability to successfully decode the source message at the relay node for high SNRs.

\subsection{RCS scheme}
The single-cell RCS scheme can be straightforwardly extended to the multi-cell scenario by applying the RCS policy at each Voronoi cell. More specifically, each AP selects the closest relay node among the relays which are located inside its circular coverage area. Equivalently to the RRS scheme, the multi-user interference affects only the broadcast phase of the cooperative protocol; therefore the steady state probability of the battery is equivalent to the single-cell case i.e., $\eta_1^{\text{RCS}}$ with $\rho=1/4\sqrt{\mu}$.
For the outage probability, we state the following theorem.

\begin{theorem}
The outage probability of the RCS scheme for the multi-cell case is given by
\begin{align}
\Pi_{\emph{RCS}_0}=\Theta(\lambda,1/4\sqrt{\mu}),
\end{align}
where $\Theta(\cdot)$ is defined in \eqref{brbr}.
\end{theorem}
\begin{proof}
The proof follows the analysis of the RCS scheme for the multi-cell case.  It can be obtained in a straightforward way by using the PDF of the closest distance given in \eqref{pdf1} in order to calculate the  probability to successfully decode the source message at the relay.  
\end{proof}

\begin{table*}
\begin{align}
\Theta(\lambda, \rho)&= \exp(-\lambda \pi \rho^2)+\bigg(1-\exp(-\lambda \pi \rho^2)\bigg) 
\nonumber \\
&\times \Bigg[ 1- \eta_1^{\text{RCS}} \underbrace{\frac{2\pi\lambda }{1-\exp\big(-\pi \lambda \rho^2 \big)} \int_{0}^{\rho}\exp\left( -\Xi x^\alpha-\lambda \pi x^2 -\frac{\pi}{16}\left[\frac{\epsilon x^\alpha}{\rho^{\alpha}}\frac{_2F_1\left(1,2\;; 2-\delta\;; \frac{1}{1+\frac{\rho^\alpha}{\epsilon x^\alpha}}\right)}{(1-\delta)\left(\frac{\epsilon x^\alpha}{\rho^\alpha}+1 \right)^2}\!-\!\frac{\epsilon x^\alpha}{\epsilon x^\alpha+\rho^\alpha}\right]  \right) x dx}_{\triangleq Q_1'(\lambda,\Xi, \rho)} \nonumber \\
&\times \frac{\lambda \exp\left(-\frac{\Xi}{\Psi} \right)}{1-\exp\big(-\pi \lambda \rho^2 \big)}\int_{0}^{2\pi} \int_{0}^{\rho} \exp\left(-\frac{\Xi}{\Psi} (r^2+\rho^2-2r \rho\cos(\theta))^{\frac{1}{\delta}}-\lambda\pi r^2 \right) r dr d\theta \Bigg], \label{brbr} 
\end{align}
\end{table*}

For the high SNR regime, the achieved outage probability can be simplified as follows

\begin{remark}
For the special case with $P\rightarrow \infty$, $P_r\rightarrow \infty$, $\Psi=P_r/P$ (constant ratio) and $N\geq 1$, the outage probability of the RCS scheme is given by
\begin{align}
\Pi_{\emph{RCS}_0^{\infty}}=1-\eta_1^{\emph{RCS}}Q_1'(\lambda,0,1/4\sqrt{\mu})
\end{align}
where $Q_1'(\cdot)$ is defined in \eqref{brbr}. 
\end{remark}

An interesting observation is that in contrast to the single-cell case, the RRS and RCS schemes do not converge to same outage floor at high SNRs.  The RCS scheme converges to a lower outage floor, since the selection of the closest relay is a mechanism to protect the source signal against multi-user interference.

\subsection{RRSB/RCSB schemes}\label{seccc}
Equivalently to the above discussion, the RRSB and RCSB schemes can be straightforwardly extended to the multi-cell scenario, by applying the RRSB and RCSB policies at each Voronoi cell. Since the generalization of the protocols does not modify the steady-state distribution of the battery and only affects the decoding probability at the relay nodes, the outage probability of the RRSB/RCSB is given by
\begin{align}
\Pi_{\text{RRSB}_0}&=\Lambda(\lambda \eta_1^{\text{RRSB}},1/4\sqrt{\mu} ), \\
\Pi_{\text{RCSB}_0}&=\Theta(\lambda \eta_1^{\text{RCSB}},1/4\sqrt{\mu}). \label{Feq}
\end{align}

By using similar arguments with the previous cases, the RRSB/RCSB schemes asymptotically converge to $\Pi_{\text{RRSB}_0}^{\infty}=\exp(-\lambda \eta_1^{\text{RRSB}})+(1-\exp(-\lambda \eta_1^{\text{RRSB}}))(1-\eta_1^{\text{RRSB}} Q_1(0,1/4\sqrt{\mu}))$ and $\Pi_{\text{RCSB}_0}^{\infty}=\exp(-\lambda \eta_1^{\text{RCSB}})+(1-\exp(-\lambda \eta_1^{\text{RCSB}}))(1-\eta_1^{\text{RCSB}} Q_1'(\lambda \eta_1^{\text{RCSB}}, 0,1/4\sqrt{\mu}))$, respectively.

\section{Numerical results}\label{NR}

Fig. \ref{fig1} plots the outage probability performance of the proposed relay selection schemes versus the transmitted power $P$. The first main observation is that the RRS and the RCS schemes converge to the same outage floor at high SNRs, as it has  been reported in Remarks \ref{cor1} and \ref{cor2}, respectively.  However, the RCS scheme slightly outperforms the RRS scheme at the moderate SNRs and thus converges to the outage floor much faster.  The RRSB and RCSB scheme, which take into account the battery status and avoid selection of uncharged relays, significantly improve the achieved performance and converge to the lowest outage floor; both schemes converge to the same outage floor at high SNRs. On the other hand, the DB scheme outperforms RRSB/RCSB schemes at low and moderate $P$. For these values, the transmission of the source signal by multiple relays through beamforming, boosts the SNR at the destination and improves the outage probability performance. In addition, it can be seen that as the density $\lambda$ increases, more relays  participate in the relaying operation and therefore the gap between DB and single-relay selection schemes increases. In the same figure, we plot the theoretical derivations given in \eqref{outag1}, \eqref{oo2}, \eqref{RRSB}, \eqref{RCSB2}, and \eqref{thr55}. The expressions in \eqref{outag1}, \eqref{oo2}, \eqref{RRSB}, \eqref{RCSB2} refer to the exact performance and perfectly match with the simulation results. On the other hand, the expression in \eqref{thr55} efficiently approximates the performance of the DB scheme; although  \eqref{thr55} holds for $\rho<<d_0$, we can see that it is a tight approximation for the considered setup with $d_0=2\rho$. 

Fig. \ref{fig2} shows the impact of the system parameters $d_0$ and $\Psi$ on the outage performance of the proposed relay selection schemes. As it can be seen in Fig. \ref{fig2}(a), as the distance between the AP and the destination increases, the outage probability increases; a larger distance corresponds to a higher path-loss degradation.  The DB scheme is more robust to the distance increase, since multiple relays cooperate through beamforming to overcome the path-loss attenuation. 
On the other hand, Fig. \ref{fig2}(b) shows that the parameter $\Psi$ is critical for the achieved performance of the system; this parameters characterizes the harvesting capability at the relay nodes  as well as the available power for relaying. As it can be seen, a small $\Psi$ facilitates the harvesting operation since the harvesting threshold is low, but the available power for relaying is not able to ensure successful decoding at the destination. If the threshold $\Psi$ is too high, the relay nodes transmit with a high power but the probability to be charged is decreased.

In Fig. \ref{revfigure}, we plot the outage probability performance for different spectral efficiencies $r_0$. As it can be seen, the convergence outage floor of the relay selection schemes is independent of the spectral efficiency. Increasing the spectral efficiency  affects only  the convergence rate of the selection schemes (slower convergence).  This observation has been expected, since according to our analysis, the convergence floor only depends on the steady-state distribution of the battery, which is not a function of the spectral efficiency. For the sake of presentation, we use small values of $r_0$ in our simulation results without loss of generality.

\begin{figure}[t]
\includegraphics[width=\linewidth]{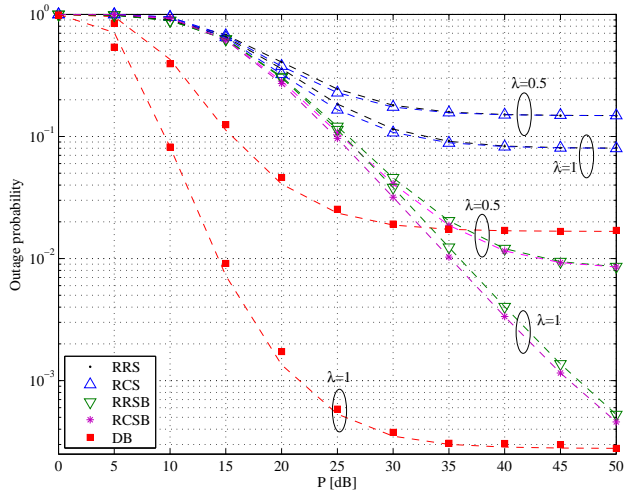}\\
\vspace{-0.3cm}
\caption{Outage probability versus $P$; $\rho=3$m, $d_0=2\rho$, $\Psi=0.1$, $\alpha=3$, $\sigma^2=1$, $r_0=0.01$ BPCU,  and $\lambda=\{1, 0.5 \}$; the dashed lines represent the theoretical results.}\label{fig1}
\end{figure}

\begin{figure}
\includegraphics[width=\linewidth]{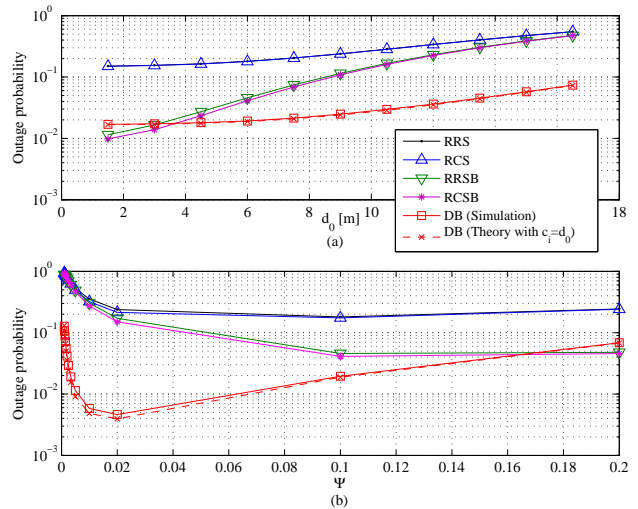}\\
\vspace{-0.3cm}
\caption{Outage probability versus a) $d_0$ and b) $\Psi$. Simulation parameters:  $P=30$ dB, $\rho=3$ m, $\alpha=3$, $\sigma^2=1$, $r_0=0.01$ BPCU, $\lambda=0.5$, a) $\Psi=0.1$, b) $d_0=2\rho$.}\label{fig2}
\end{figure}

 \begin{figure}
 \centering
\includegraphics[width=\linewidth]{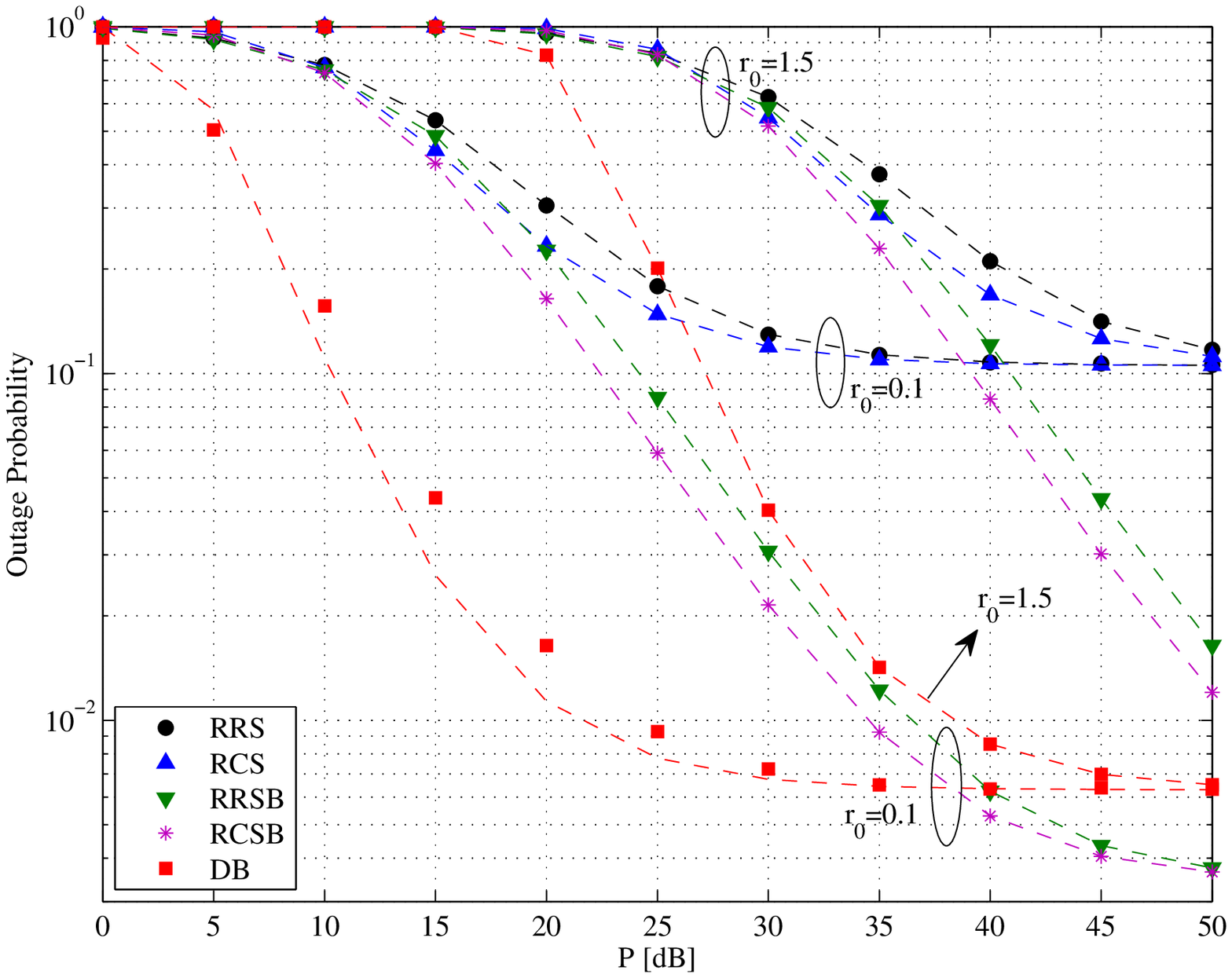}\\
\vspace{-0.3cm}
\caption{Outage probability versus $P$ for different values of $r_0$. Simulation parameters $d_0=\rho$, $\Psi=0.1$, $\alpha=3$, $\sigma^2=1$, $r_0=\{0.1,1.5\}$ BPCU, $\lambda=1$; the dashed lines represent the theoretical results.}\label{revfigure}
\end{figure}

\begin{figure}
\includegraphics[width=\linewidth]{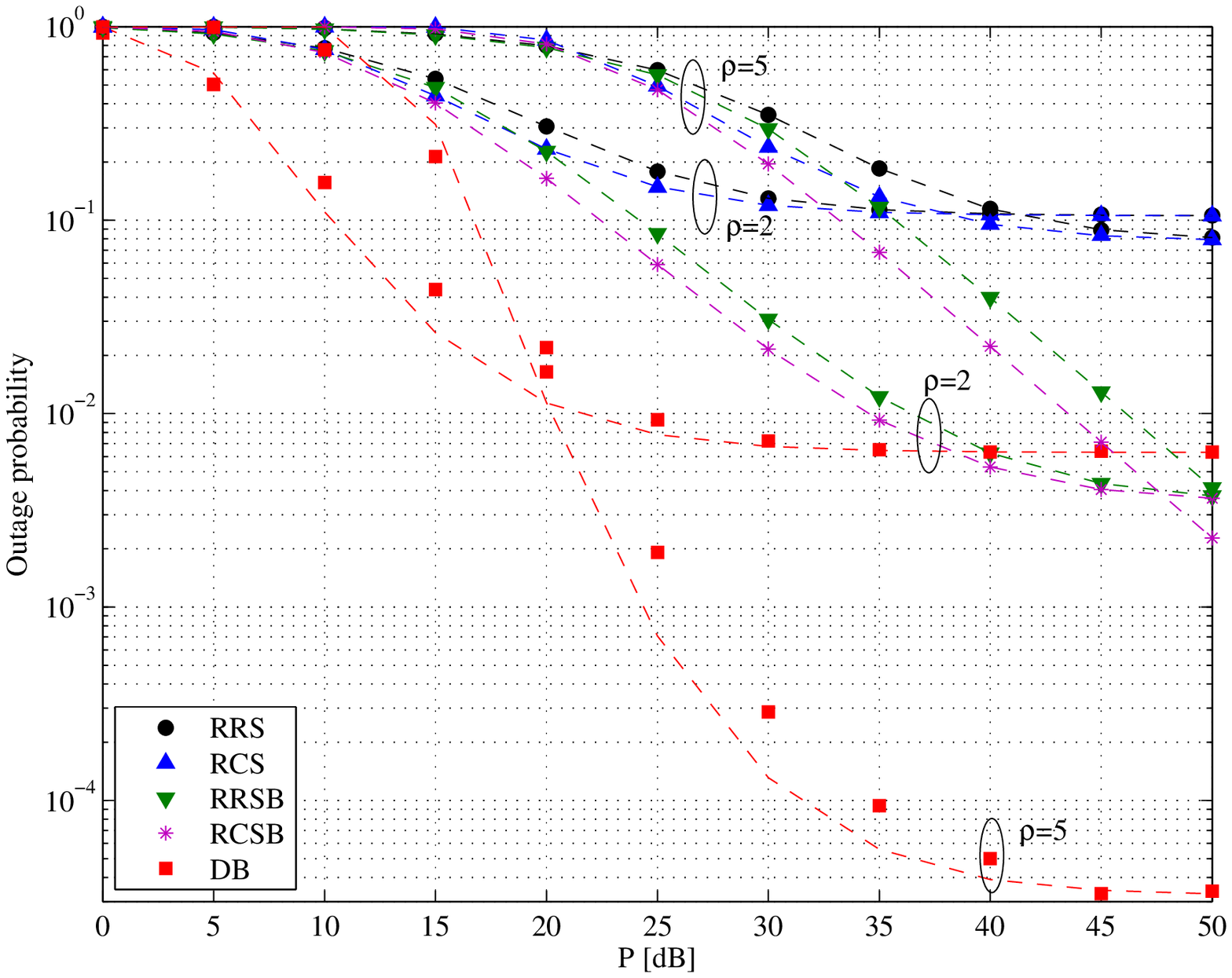}\\
\vspace{-0.3cm}
\caption{Outage probability versus $P$ for different values of $\rho$. Simulation parameters: $d_0=\rho$, $\Psi=0.1,$ $\alpha=3$, $\sigma^2=1$, $r_0=0.1$ BPCU, $\lambda=1$; the dashed lines represent the theoretical results.}\label{fig3}
\end{figure}

\begin{figure}
\includegraphics[width=\linewidth]{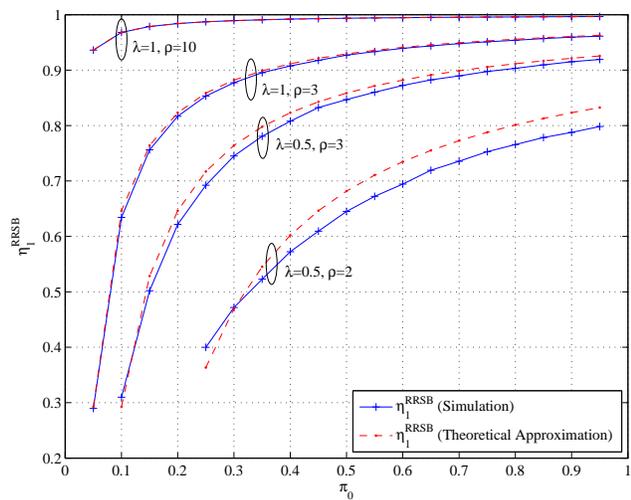}\\
\vspace{-0.3cm}
\caption{Steady state-distribution of the battery for the RRSB/RCSB schemes;  efficiency of the proposed approximation in \eqref{aprjens} for different settings ($\lambda, \rho$).}\label{Approximation}
\end{figure}

Fig. \ref{fig3} shows the impact of the radius $\rho$ on the achieved outage probability performance. The first main observation is that as $\rho$ increases, the performance of the proposed relay selection schemes is improved at high SNRs i.e., they converge to lower outage floors. This observation was expected, since as the area of the disc increases, the number of the relay nodes increases (with $\mathbb{E}[N]=\lambda \pi \rho^2$) and thus  a) more relays participate in the relaying operation (DB scheme), b) the probability to select an uncharged relay node is decreased (RRS, RCS schemes), c) the probability to have no charged relay in the disc is decreased (RRSB, RCSB schemes). However, for the low and moderate SNRs, it can be seen that as $\rho$ increases, the outage probability performance decreases. In this SNR regime, the outage probability is dominated by the channel path-loss attenuation and becomes more severe as $\rho$ increases; a larger $\rho$ corresponds to a higher distance ($d_0=\rho$) between source-destination and therefore increases the longest hop link.

Fig. \ref{Approximation} deals with the steady-state distribution of the battery in the RRSB and RCSB schemes. The closed form expression given in Proposition \ref{cor4} is based on the transition probability  $\pi_1^{\text{RRSB}}$, which uses the Jensen's approximation (inequality); see the  proof of Proposition \ref{cor4}. In order to show the efficiency of this approximation, Fig. \ref{Approximation} compares the exact $\eta_1^{\text{RRSB}}$ (given by simulation results) against the proposed approximation for different system parameters. As it can be seen, the proposed expression efficiently approximates the steady-state distribution and its accuracy is improved as the term $\pi_0-1/\lambda\pi \rho^2$ increases. Therefore, this approximation allows closed-form expression for the steady-state distribution and provides accurate results for the scenarios of interest i.e., $\lambda \pi \rho^2>>0$.   

\begin{figure}[t]
\includegraphics[width=\linewidth]{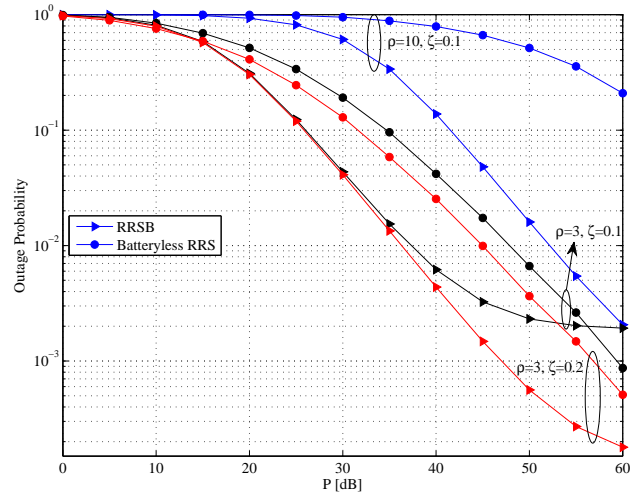}\\
\vspace{-0.3cm}
\caption{Outage probability for the RRSB and the batteryless RRS \cite{DIN1} schemes versus $P$. Simulation parameters: $d_0=\rho$, $\Psi=0.02,$ $\alpha=3$, $\sigma^2=1$, $r_0=0.01$ BPCU, $\lambda=1$.}\label{RRSfig}
\end{figure}

\begin{figure}
\includegraphics[width=\linewidth]{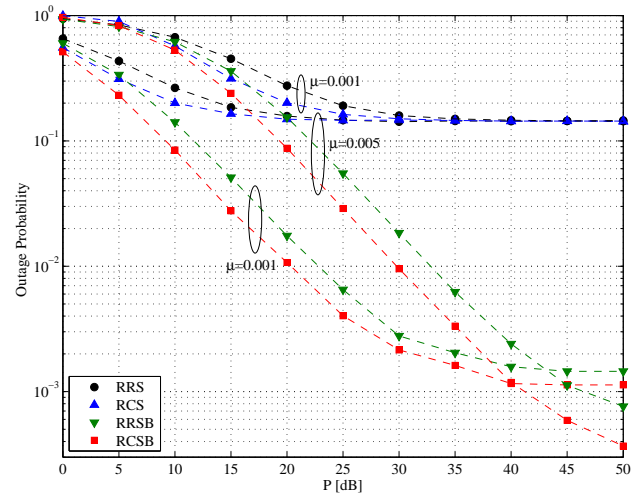}\\
\vspace{-0.3cm}
\caption{Outage probability versus $P$. Simulation parameters: $d_0=\rho$, $\Psi=0.1,$ $\alpha=3$, $\sigma^2=1$, $r_0=0.001$ BPCU, $\lambda=0.5$, and $\mu=\{0.001,0.005\}$; the dashed lines represent the theoretical results.}\label{fig10}
\end{figure}

In Fig. \ref{RRSfig}, we compare the proposed RRSB scheme with the batteryless RRS scheme proposed in \cite{DIN1}. The batteryless RRS scheme uses the PS-SWIPT technique at the relay node and the harvested energy is directly used to power the relaying link. It can be seen, that the RRSB scheme outperforms batteryless RRS at low and moderate SNRs and the gain increases as the conversion efficiency $\zeta$ decreases and/or the disc radius increases. For this SNR regime (which is expanded as $\zeta$ decreases), the PS-SWIPT technique becomes inefficient and the integration of the battery, which decouples the information and power transfer in time, significantly improves the performance. For high SNRs, the RRSB scheme suffers from an outage floor, while batteryless RRS provides a diversity gain equal to one \cite{DIN1}.  It is worth noting that the batteryless RRS scheme assumes perfect channel knowledge at the relay node, and requires appropriate electronic circuits in order to perform dynamic PS;  therefore it corresponds to a higher complexity than the proposed RRSB scheme.        

Fig. \ref{fig10} deals with the application of the proposed single-relay selection schemes to a multi-cell scenario with multi-user interference at the relay nodes (Section \ref{SEC3}). The RCS scheme outperforms the RRS scheme for the low/intermediate SNRs, while both schemes converge to similar outage probability floors; for this simulation setup, the success probabilities at the relay nodes  become almost equal at high SNRs  for both $\mu$ values. The consideration of the battery status into relay selection significantly improves the outage performance of the system and thus RRSB and RCSB schemes achieve lower outage probability floors. It is worth noting that RRSB and RCSB schemes do not converge to the same outage floor at high SNRs, since the decoding probability at the relays is different for the two schemes and affects their convergence (see Section \ref{seccc}). On the other hand, as the AP's density $\mu$ increases, the radius of the Voronoi cells decrease and the curves follow the discussion in Fig. \ref{fig3}. Theoretical results perfectly match with the simulation curves and validate our analysis.

\section{Conclusion}\label{SEC5}
In this paper, we have studied the problem of relay selection in WPC cooperative networks with spatially random relays. We assume that the relay nodes are equipped with batteries and use the received signal either for conventional decoding or battery charging. Based on a single-cell network topology, we investigate several relay selection schemes with different complexities. Their outage probability performance is derived in closed form by modeling the behavior of the battery as a two-state MC. We prove that the relay selection schemes suffer from an outage floor at high SNRs, which highly depends on the steady-state distribution of the battery. The RRS and RCS schemes achieve the worst outage probability performance and  converge to the same outage floor, while the consideration of the battery status significantly improves their achieved performance. The DB scheme is a promising solution for low/moderate SNRs and outperforms single-relay selection at the cost of a CSI. The proposed selection schemes are generalized to multi-cell network topologies, where multi-user interference affects relay decoding.  

\appendix

\subsection{Proof of Proposition \ref{thr1}}\label{APB}

Firstly, we calculate the cumulative distribution function (CDF) of the random variable $u_i \triangleq \frac{|h_i|^2}{1+d_i^\alpha}$; this result is essential for the derivation of the steady-state distribution. From the system model,  $|h_i|^2$ is an exponential random variable with unit variance and thus its CDF is equal to $F_{h}(x)=1-\exp(-x)$. The CDF of the random variable $u_i$ is given as follows 
\begin{subequations}
\label{pi000}
\begin{align}
F_{u}(x)&=\mathbb{P}\{u_i<x \} \nonumber \\
&=1-\mathbb{P}\left\{|h_i|^2> x(1+d_i^\alpha) \right\} \nonumber \\
&=1-\mathbb{E}\exp \left\{- x(1+d_i^\alpha)  \right\} \nonumber \\
&=1-\int_{\mathcal{D}}\exp \left\{- x(1+y^\alpha) \right\}f_{d}(y)dy \label{pi0} \\
&=1-\frac{1}{\pi \rho^2}\int_{0}^{2\pi} \int_{0}^{\rho}\exp \left\{- x(1+y^\alpha)  \right\}y dyd\theta \nonumber \\
&=1-\frac{2}{\rho^2}\exp \left(-x\right)\int_{0}^{\rho}y \exp \left(-x y^\alpha \right) dy \nonumber \\
&=1-\frac{\delta}{\rho^2}\exp \left(-x\right)\frac{\gamma(\delta, x \rho^\alpha)}{x^{\delta}}, \label{pi00}
\end{align}
\end{subequations}
where $f_{d}(x)=1/\pi\rho^2$ in \eqref{pi0} denotes the probability density function (PDF) of each point in the disk $\mathcal{D}$ and the result in \eqref{pi00} is based on \cite[Eq. 3.381.8]{GRAD}.  

By using the CDF of the random variable $u_i$ in \eqref{pi000}, the probability that an empty battery is fully charged during the broadcast phase can be expressed as  
\begin{align}
\pi_0=\mathbb{P}\left\{u_i> \frac{P_r}{P} \right\}=\frac{\delta}{\rho^2}\exp \left(-\Psi \right)\frac{\gamma(\delta, \Psi \rho^\alpha)}{\Psi^{\delta}}. \label{po}
\end{align}

On the other hand, if $N$ is the number of relays in $\mathcal{D}$, the probability to select a relay according to the RRS policy is equal to $1/N$; the probability that a charged relay becomes uncharged is equal to this selection probability and can be expressed as
\begin{align}
\pi_1^{\text{RRS}}&=\mathbb{E}\left[\frac{1}{N} \right]=\exp(-\lambda \pi \rho^2)\sum_{k=1}^{\infty} \frac{(\lambda \pi \rho^2)^k}{k!}\cdot \frac{1}{k} \nonumber \\
&\geq \frac{1}{\mathbb{E}[N]}=\frac{1}{\lambda \pi \rho^2}, \label{pi1}
\end{align} 
where \eqref{pi1} is based on Jensen's inequality and $\mathbb{E}[N]=\lambda \pi \rho^2$ is the average number of relays in $\mathcal{D}$ (from the definition of a PPP). It is worth noting that the probability $\pi_1^{\text{RRS}}$ is simplified by using Jensen's inequality; this approximation significantly simplifies our derivations and is tight as $\lambda \pi \rho^2$ increases. By combining \eqref{po} and \eqref{pi1} and substituting back into \eqref{sst}, the expression in Proposition \ref{thr1} is proven.

\subsection{Proof of Theorem \ref{thr2}}\label{APC}

The outage probability for the RRS scheme can be expressed as 
\begin{align}
\Pi_{\text{RRS}}&=\mathbb{P}\{N=0\}+\mathbb{P}\{N\geq 1, S(R_i)=s_0 \} \nonumber \\
&\;+\mathbb{P}\{N\geq 1, S(R_i)=s_1, {\sf SNR}_i<\epsilon \} \nonumber \\
&\;+\mathbb{P}\{N\geq 1, S(R_i)=s_1, {\sf SNR}_i \geq \epsilon, {\sf SNR}_D< \epsilon \} \nonumber \\
&=\mathbb{P}\{N=0\}+\mathbb{P}\{N\geq 1\}(1-\eta_1^{\text{RRS}})+\mathbb{P}\{N\geq 1\}\eta_1^{\text{RRS}} \nonumber \\
&\;\times \bigg(1-\underbrace{\mathbb{P}\{{\sf SNR}_i\geq \epsilon|N\geq 1\}\mathbb{P}\{{\sf SNR}_D\geq \epsilon|N\geq1\}}_{Q}\bigg) \nonumber \\
&=\mathbb{P}\{N=0\}+\mathbb{P}\{N\geq 1\}(1-\eta_1^{\text{RRS}}Q). \label{outage1}
\end{align}
From the Poisson distribution, we have
\begin{align} 
\mathbb{P} \{N=0\}=\exp(-\lambda \pi \rho^2). \label{out3}
\end{align}
In addition, 
\begin{align}
\mathbb{P}\{{\sf SNR}_i\geq \epsilon| N\geq 1\} &=\mathbb{P}\{u_i\geq \Xi \}=1-F_u(\Xi), \label{out1}
\end{align}
where $F_u(\cdot)$ is given in \eqref{pi000}. Finally, the success probability for the relaying link is written as
\begin{subequations}
\label{out2}
\begin{align}
&\mathbb{P}\{{\sf SNR}_D\geq \epsilon | N\geq 1 \}=\mathbb{P} \left\{|g_i|^2 \geq \frac{\epsilon  (1+c_i^\alpha)}{P_r} \right\} \nonumber \\
&=\mathbb{E}\exp\left(-\frac{\Xi (1+c_i^\alpha)}{\Psi} \right) \nonumber \\
&=\mathbb{E}\exp\left(-\frac{\Xi}{\Psi} \left[1+ (d_i^2+d_0^2-2d_i d_0\cos (\theta))^\frac{1}{\delta} \right]\right)  \label{cos_rule}\\
&=\int_{\mathcal{D}}\exp\left(-\frac{\Xi}{\Psi}\left[1+(x^2+d_0^2-2x d_0\cos (\theta))^\frac{1}{\delta}\right] \right)f_{d_i}(x)dx \nonumber \\
&=\frac{1}{\pi \rho^2}\exp\left(-\frac{\Xi}{\Psi}\right) \nonumber \\
&\;\;\;\times \int_{0}^{2\pi}\int_{0}^{\rho}\exp\left(-\frac{\Xi}{\Psi} (x^2+d_0^2-2x d_0\cos (\theta))^\frac{1}{\delta} \right)x dx d\theta,
\end{align}
\end{subequations}
where $c_i^2=d_i^2+d_0^2-2d_i d_0 \cos(\theta)$ in \eqref{cos_rule} holds by using the cosine law. By combining \eqref{out3}, \eqref{out1}, \eqref{out2} and substituting back into \eqref{outage1}, we prove Theorem \ref{thr2}.

\subsection{Proof of Theorem \ref{thr3}}\label{APD}

The PDF of the nearest  distance $d_{i^*}$ for the homogeneous PPP $\Phi$ with intensity $\lambda$, conditioned on $N\geq 1$, is given by \cite[Eq. (33)]{DIN2}
\begin{align}\label{pdf1}
f_r(r)=\frac{2 \lambda \pi}{1-\exp\big(-\lambda  \pi\rho^2 \big)} r \exp\big(-\lambda \pi r^2 \big).
\end{align}

The probability to successfully decode the source message at the selected relay, conditioned on $N\geq 1$, can be expressed as follows
\begin{align}
&\mathbb{P} \left\{{\sf SNR}_i\geq \epsilon \big|N\geq 1 \right\}=\mathbb{P} \left\{P\frac{|h_{i^*}|^2}{(1+d_{i^*}^\alpha)\sigma^2} \geq \epsilon \right\} \nonumber \\
&=\mathbb{P}\left\{|h_{i^*}|^2\geq \Xi(1+d_{i^*}^\alpha) \right\} \nonumber \\
&=\mathbb{E}\left\{\exp\left(-\Xi(1+d_{i^*}^\alpha)  \right) \right\} \nonumber \\
&=\int_{0}^{\rho}\exp\left(-\Xi(1+r^\alpha)  \right)f_r(r)dr \nonumber \\
&=\frac{2\pi\lambda \exp\left(-\Xi \right)}{1-\exp\big(-\pi \lambda \rho^2 \big)} \int_{0}^{\rho} \exp\bigg( -\Xi r^\alpha-\lambda \pi r^2  \bigg) r dr. \label{v1}
\end{align}  

The probability to successfully decode the relaying signal at the destination, conditioned on $N\geq 1$, is expressed as 
\begin{subequations}
\begin{align}
&\mathbb{P} \left\{{\sf SNR}_D\geq \epsilon \big| N\geq 1 \right\}=\mathbb{P} \left\{|g_{i^*}|^2 \geq \frac{\Xi (1+c_{i^*}^\alpha)}{\Psi} \right\} \nonumber \\
&=\mathbb{E}\left\{\exp\left(-\frac{\Xi (1+c_{i^*}^\alpha)}{\Psi} \right) \right\} \label{v2v2} \\
&=\frac{\lambda \exp\left(-\Xi/\Psi \right)}{1-\exp\big(-\lambda \pi \rho^2 \big)} \nonumber \\
&\times\!\! \int_{0}^{2\pi}\!\!\!\!\int_{0}^{\rho}\!\!\!\exp\!\!\left(-\frac{\Xi}{\Psi} (r^2+d_0^2-2r d_0\cos(\theta))^{\frac{1}{\delta}}-\lambda\pi r^2 \right) r dr d\theta. \label{v2}
\end{align}
\end{subequations}

The outage probability for the RCS scheme can be expressed by the general expression in \eqref{outage1}. By combining  \eqref{eta1}, \eqref{out3}, \eqref{v1}, \eqref{v2} and substituting  into \eqref{outage1}, we prove the statement in Theorem \ref{thr3}.

\subsection{Proof of Remark \ref{cor3}}\label{APD2}

For the special case that $\alpha=2$, \eqref{v1} is simplified to 
\begin{subequations}
\label{en12}
\begin{align}
&\mathbb{P} \left\{{\sf SNR}_i\geq \epsilon \big|N\geq 1 \right\} \nonumber \\
&=\frac{2\pi\lambda \exp\left(-\Xi \right)}{1-\exp\big(-\pi \lambda \rho^2 \big)}\int_{0}^{\rho} r\exp\left(-(\Xi+\lambda\pi)r^2 \right)dr \nonumber \\
&=\frac{\lambda \pi \exp\left(-\Xi \right)}{1-\exp\big(-\pi \lambda \rho^2 \big)}\cdot \frac{1-\exp\big(-(\Xi+\lambda \pi)\rho^2 \big)}{(\Xi+\lambda \pi)} \label{en1} \\
&\approx \frac{\lambda\pi(1-\Xi)}{\lambda \pi+\Xi} \left(1+\Xi\rho^2 \frac{\exp(-\lambda \pi \rho^2)}{1-\exp(-\lambda \pi \rho^2)} \right)\;\;\;\;\;\text{(for $\Xi\rightarrow 0$)}, \label{en2}
\end{align}
\end{subequations}
where \eqref{en1} holds from \cite[Eq. 3.381.8]{GRAD}, the asymptotic expression in \eqref{en2} uses the approximation $1-\exp(-x)\approx x$ for $x\rightarrow 0$. 

As for the relaying link, for the special case that $\rho<<d_0$, the Euclidean distance of the link relay-destination becomes $c_{i^*}\approx d_0$ and the expression in \eqref{v2v2} is simplified to 
\begin{align}
\mathbb{P} \left\{{\sf SNR}_D\geq \epsilon \big| N\geq 1 \right\}&=\exp\left(-\frac{\Xi (1+d_0^\alpha)}{\Psi}\right) \nonumber \\
&\approx 1-\frac{\Xi(1+d_0^\alpha)}{\Psi}\;\;\;\;\text{(for $\Xi\rightarrow 0$)}. \label{en3}
\end{align}
By combining  \eqref{en2}, \eqref{en3} and substituting into \eqref{outage1} with $\mathbb{P}\{N=0\}=0$, we prove Remark \ref{cor3}.

\subsection{Proof of Theorem \ref{Thr5}}\label{APE}

In the DB scheme, the relay nodes which are active (charged batteries) and are able to decode the source message, participate in the relaying transmission. The participating relays ($i \in \mathcal{C}$) form a PPP $\Phi'$, which yields from the original PPP $\Phi$ by applying an independent thinning operation. More specifically, the PPP $\Phi'$ has a density 
\begin{align}
\lambda'=\lambda \eta_1^{\text{DB}}(1-F_u(\Xi)).\label{ll}
\end{align}

In order to simplify the analysis, we focus on a disk $\mathcal{D}$ with $\rho<<d_0$; in this case, the distance between each relay and the destination becomes equal to $d_0$ i.e., $c_i\approx d_0$. The SNR expression in \eqref{memoides} is simplified to 
\begin{align}
{\sf SNR}_D=\frac{P_r}{(1+d_0^\alpha)\sigma^2}\sum_{i\in \mathcal{C}}|g_i|^2.
\end{align} 
The random variable $Y=\sum_{i\in \mathcal{C}}|g_i|^2$ is the sum of $|\mathcal{C}|=K$ independent and identically distributed (i.i.d.) exponential random variables, where the cardinality $K$ follows a Poisson distribution with density $\lambda'$. Therefore, the outage probability of the system is given by: 
\begin{align}
\Pi_{\text{DB}}&=\mathbb{P}\{{\sf SNR}_D<\epsilon \} \nonumber \\
&=\mathbb{P}\bigg\{Y< \underbrace{\frac{\Xi(1+d_0^\alpha)}{\Psi}}_{\triangleq  z_0} \bigg\} \nonumber \\
&=\sum_{k=0}^{\infty}\mathbb{P}\{Y< z_0|K=k\}\mathbb{P}\{ K=k\} \nonumber \\
&=\sum_{k=0}^{\infty} \frac{\gamma(k,z_0)}{\Gamma(k)} \exp(-\lambda' \pi \rho^2)  \frac{(\lambda' \pi \rho^2)^k}{k!}, \label{out5}
\end{align}
where $F_Y(x,k)=\frac{\gamma(k,x)}{\Gamma(k)}$ denotes the CDF of the random variable $Y$ i.e., Gamma distribution with shape parameter $k$. By plugging \eqref{ll} into \eqref{out5}, we prove the expression in Theorem \ref{Thr5}.

\subsection{Proof of Remark \ref{cor4}}\label{APF}

In the case that  $z_0\rightarrow 0$ (i.e., high SNRs with $P\rightarrow \infty$), the expression in \eqref{out5} is simplified to 
\begin{subequations}
\begin{align}
\Pi_{\text{DB}}^{\infty}&\approx \exp(-\lambda' \pi \rho^2) \sum_{k=0}^{\infty} \frac{(z_0)^k (\lambda' \pi \rho^2)^k}{(k!)^2} \label{q1} \\
&=\exp(-\lambda' \pi \rho^2) I_0 \left(2\rho \sqrt{ z_0 \lambda' \pi}  \right) \label{q2} \\
&\rightarrow \exp(-\lambda \eta_1^{\text{DB}} \pi \rho^2), \label{q3}
\end{align}
\end{subequations}
where $I_0(\cdot)$ denotes the zeroth order Modified Bessel function of the first kind, \eqref{q1} uses the approximation $\gamma(a,x)\approx x^a/a$ for small $x$, \eqref{q2} is based on the series representation in  \cite[Eq. 8.447.1]{GRAD}, and \eqref{q3} uses the approximation $I_0(x)\approx 1$ for small $x$.

\subsection{Proof of Theorem \ref{Thr6}}\label{APH}

The outage probability of the RRS scheme follows the analysis of the single-cell case, presented in Appendix \ref{APC}. The main difference is the computation of the decoding probability at the relay node. This probability conditioned on $d_i$ (distance between AP and relay),  can be expressed as follows
\begin{align}
\mathbb{P}\{{\sf SINR}_i\geq \epsilon | N\geq 1, d_i \}&=\mathbb{P} \left\{\frac{\frac{P |h_i|^2}{d_i^\alpha}}{P I+\sigma^2}\geq \epsilon  \right\} \nonumber \\
&=\mathbb{P} \left \{|h_i|^2 \geq \Xi d_i^\alpha+ \epsilon d_i^\alpha I \right\} \nonumber \\
&=\exp\left(-\Xi d_i^\alpha \right)\exp\left(-\epsilon d_i^\alpha I \right) \nonumber \\
&= \exp\left(-\Xi d_i^\alpha\right) L\big(\epsilon d_i^\alpha \big),
\end{align}
where $L(\cdot)$ is the Laplace transform of the random variable $I=\sum_{j\in \Upsilon/\{\mathcal{D} \}} \frac{H_j}{r_j^\alpha}$. With expectation over $d_i$, we obtain 
\begin{align}
\mathbb{P}\{{\sf SNR}_i\geq \epsilon | N\geq 1\}\!&=\!\!\!\int_{0}^{2\pi}\!\!\! \int_{0}^{\rho}\!\!\!\exp\left(-\Xi x^\alpha \right) L\big(\epsilon x^\alpha \big) f_d(x)xdxd\theta \nonumber \\
&=\frac{2}{\rho^2}\int_{0}^{\rho}\!\!\exp\left(-\Xi x^\alpha \right) L\big(\epsilon x^\alpha \big) x dx, \label{teliko}
\end{align}
where $f_d(x)=1/\pi \rho^2$ denotes the PDF of each point in the (typical) disc. 

The Laplace transform of the interference $I$ can be calculated as follows 
\begin{align} \label{Laplace}
L(s)&=\mathbb{E}\exp(-s I) \nonumber \\
&=\mathbb{E}\left(\prod_{j\in \Upsilon/\{\mathcal{D} \}}\exp(-s H_j r_j^{-\alpha}) \right) \nonumber \\
&= \exp\left\{ - 2\pi \mu \int_{\rho}^{\infty}\mathbb{E}_H \bigg[1-\exp(-sH r^{-\alpha}) \bigg]r dr \right\}.
\end{align}
where \eqref{Laplace} uses the probability generating functional of a PPP \cite[Sec. 4.6]{HAN}. Conditioned on $H$, we have
\begin{subequations}

\begin{align}
&2\int_{\rho}^{\infty}\big(1-\exp(-sH r^{-\alpha}) \big)r dr \nonumber \\
&= \delta \int_{\rho^\alpha}^{\infty} \left(1-\exp\left(-\frac{sH}{y} \right) \right)y^{\delta-1} dy \label{x1} \\
&=\delta \int_{0}^{\frac{1}{\rho^\alpha}}\bigg(1-\exp(-sHx) \bigg) x^{-\delta-1} dx \label{x2}  \\
&=\left(\exp\left(-\frac{sH}{\rho^\alpha} \right)-1\right)\rho^2+sH\int_{0}^{\frac{1}{\rho^\alpha}}x^{-\delta}\exp(-sHx)dx \label{x3} \\
&=\left(\exp\left(-\frac{sH}{\rho^\alpha} \right)-1\right)\rho^2+(sH)^{\delta} \gamma\left(1-\delta, \frac{sH}{\rho^\alpha} \right), \label{x4} 
\end{align}
\end{subequations}
where \eqref{x1} follows from the substitution $y\leftarrow r^{\alpha}$, \eqref{x2} from the substitution $x\leftarrow y^{-1}$, \eqref{x3} from integration by parts, and \eqref{x4} is based on \cite[Eq. 3.381.8]{GRAD}. With the expectation over $H$, we have
\begin{align}
&\int_{0}^{\infty} \rho^2 \left(\exp\left(-\frac{sH}{\rho^\alpha} \right)-1 \right)\exp(-H)dH \nonumber \\
&+s^{\delta} \int_{0}^{\infty}H^{\delta}\gamma\left(1-\delta,\frac{sH}{\rho^\alpha}\exp(-H) \right)dH \nonumber \\
&=s^{\delta} \left(\frac{s}{\rho^\alpha}\right)^{1-\delta}\frac{_2F_1\left(1,2\;; 2-\delta\;; \frac{\frac{s}{\rho^\alpha}}{\frac{s}{\rho^\alpha}+1}\right)}{(1-\delta)\left(\frac{s}{\rho^\alpha}+1 \right)^2}- \rho^2\frac{s}{s+\rho^\alpha}, \label{x5}
\end{align}
where \eqref{x5} is based on the expression in\cite[Eq. 6.455.2]{GRAD}. By substituting \eqref{x5} into \eqref{Laplace}, we have
\begin{align}
L(s)\!\!=\!\exp\!\left(\!\!-\pi \mu\!\left[\frac{s}{\rho^{\alpha-2}}\frac{_2F_1\left(1,2\;; 2-\delta\;; \frac{\frac{s}{\rho^\alpha}}{\frac{s}{\rho^\alpha}+1}\right)}{(1-\delta)\left(\frac{s}{\rho^\alpha}+1 \right)^2}\!-\!\frac{\rho^2s}{s+\rho^\alpha}\right]  \!\right).\label{Laplace2}
\end{align}

By combining \eqref{eta1}, \eqref{out3}, \eqref{out2}, \eqref{teliko} and substituting back  into \eqref{outage1} with $d_0=\rho=1/4\sqrt{\mu}$ we prove Theorem \ref{Thr6}.

\begin{IEEEbiography}[{\includegraphics[width=1in,height=1.25in,clip,keepaspectratio]{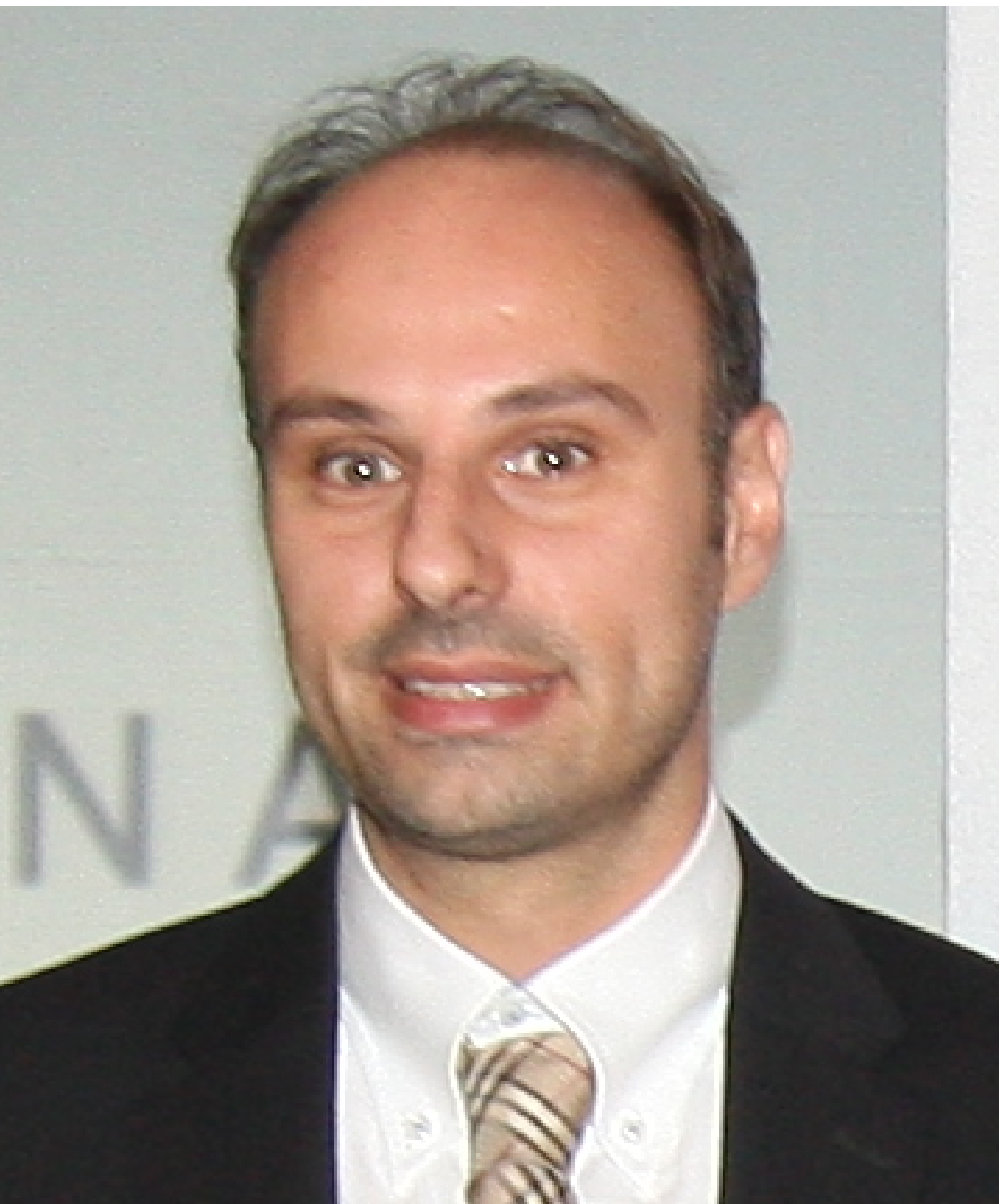}}]{Ioannis Krikidis} (S'03-M'07-SM'12) received the diploma in Computer Engineering from the Computer Engineering and Informatics Department (CEID) of the University of Patras, Greece, in 2000, and the M.Sc and Ph.D degrees from Ecole Nationale Sup\'erieure des T\'el\'ecommunications (ENST), Paris, France, in 2001 and 2005, respectively, all in electrical engineering. From 2006 to 2007 he worked, as a Post-Doctoral researcher, with ENST, Paris, France, and from 2007 to 2010 he was a Research Fellow in the School of Engineering and Electronics at the University of Edinburgh, Edinburgh, UK. He has held also research positions at the Department of Electrical Engineering, University of Notre Dame; the Department of Electrical and Computer Engineering, University of Maryland; the Interdisciplinary Centre for Security, Reliability and Trust, University of Luxembourg; and the Department of Electrical and Electronic Engineering, Niigata University, Japan. He is currently an Assistant Professor at the Department of Electrical and Computer Engineering, University of Cyprus, Nicosia, Cyprus. His current research interests include  communication theory, wireless communications, cooperative networks, cognitive radio and secrecy communications.

Dr. Krikidis serves as an Associate Editor for IEEE Transactions on Communications,  IEEE Transactions on Vehicular Technology, IEEE Wireless Communications Letters, and Wiley Transactions on Emerging Telecommunications Technologies.  He was the Technical Program Co-Chair for the IEEE International Symposium on Signal Processing and Information Technology 2013. He received an IEEE Communications Letters and IEEE Wireless Communications Letters exemplary reviewer certificate in 2012. He was the recipient of the {\it Research Award Young Researcher} from the Research Promotion Foundation, Cyprus, in 2013.
\end{IEEEbiography}

\end{document}